%% file: main.tex
\documentclass{article}
\usepackage{kr}

\usepackage{times}
\usepackage{soul}
\usepackage{url}
\usepackage[hidelinks]{hyperref}
\usepackage[utf8]{inputenc}
\usepackage[small]{caption}
\usepackage{graphicx}
\usepackage{amsmath}
\usepackage{amssymb}
\usepackage{amsthm}
\usepackage{booktabs}
\usepackage{algorithm}
\usepackage{algorithmic}
\usepackage{calc}
\newtheorem{example}{Example}
\newtheorem{theorem}{Theorem}
\newtheorem{corollary}{Corollary}

\newtheorem{definition}{Definition}
\usepackage{esvect}
\usepackage{macros}
\usepackage{enumitem}
\newenvironment{sketch}[1][\proofname]{\proof[#1]\mbox{}}{\endproof}

\title{Repairing General Game Descriptions (extended version)}

\author{%
Yifan He$^1$\and
Munyque Mittelmann$^2$\and
Aniello Murano$^3$\and
Abdallah Saffidine$^4$ \and
Michael Thielscher$^1$\\
\affiliations
$^1$University of New South Wales, Australia \\
$^2$CNRS, LIPN, Sorbonne Paris North University, France\\
$^3$University of Naples Federico II, Italy\\
$^4$Potassco Solutions, Germany\\
\emails
\{yifan.he1,mit\}@unsw.edu.au, mittelmann@lipn.univ-paris13.fr, 
aniello.murano@unina.it, abdallahs@gmail.com
}

\begin{document}

\maketitle
\input{example-gdl}
\begin{abstract}
% Although game descriptions in General Game Playing (GGP) are required to be well-formed, the GGP specification language permits the definition of games that violate this requirement. 
%    In this paper, we consider the problem of finding minimal repairs for game descriptions that do not comply with a formal requirement. 
%    We provide tight complexity results for different computational problems related to game repair. We then present an approach based on Answer Set Programming to solve the minimal repair problem and demonstrate its utilization in ill-defined game descriptions
    The Game Description Language (GDL) is a widely used formalism for specifying the rules of general games. Writing correct GDL descriptions can be challenging, especially for non-experts. Automated theorem proving has been proposed to assist game design by verifying if a GDL description satisfies desirable logical properties. However, when a description is proved to be faulty, the repair task itself can only be done manually. Motivated by the work on repairing unsolvable planning domain descriptions, we define a more general problem of finding minimal repairs for GDL descriptions that violate formal requirements, and we provide complexity results for various computational problems related to minimal repair. Moreover, we present an Answer Set Programming-based encoding for solving the minimal repair problem and demonstrate its application for automatically repairing ill-defined game descriptions.
\end{abstract}

\input{intro2}
\input{preliminary}
\input{problem-definition}

\input{encoding-proof}
\input{experiment-2}

%\input{experiments}
\input{conclusion}

\bibliographystyle{kr}
\bibliography{kr-sample}

\appendix
\input{appendix}
\end{document}

%% file: example-gdl.tex
\begin{figure*}
\footnotesize
\begin{center}
\begin{verbatim}
[cr] role(p).    base(win).  base(loss).      terminal:- true(win). goal(p,100):- true(win).
     input(p,l). input(p,r).                  terminal:- true(loss). goal(p, 0):- true(loss).
[c1] legal(p,l). [c2] next(loss):- does(p,l). [c3] next(win):- does(p,r).
\end{verbatim}
\end{center}
\caption{GDL description of a simple game. $c_{1},..,c_{r}$ are rule numbers, with $c_{r}$ containing all rules other than $c_{1},c_{2}$ and $c_{3}$. $c_{1}$ says $l$ is a legal action for $p$. $c_{2}$ says that $\gdl{loss}$ will always hold after ~$p$ does $l$, while $c_{3}$ says that $\gdl{win}$ would hold if~$p$ could do $r$ instead.}
%$c_{r}$ defines the terminal, goal relation, . 
%\munyque{What does the symbols before the rules mean (e.g. c1, c2,..., cr)? Is it the rule number?} 
%\munyque{Instead of saying "The game has a rule $legal(p,r)$ missing" we could say in the end:
%"The game is not weakly winnable by $p$, because the action $r$ is not legal for her." 
%}
%\yifan{I put the explanation of the game at the end section 2.1}
%\yifan{fixed}
\label{fig:gdl}
\end{figure*}

%% file: intro2.tex
\section{Introduction}
%\yifanm{General game players are systems that can play strategic games
%based on game descriptions supplied at runtime. These systems are
%generally unaware of the game rules until the competition starts.
%GDL, the general Game Description Language, is a widely used language for describing the rules of games~\cite{genesereth2005general}. }

The Game Description Language (GDL) has been developed
as a lightweight knowledge representation formalism for describing the rules of arbitrary finite games~\cite{genesereth2005general}. It is used as the input language for general game-playing (GGP) systems, which
can learn to play any new game from the mere rules without human intervention, thus exhibiting a form of general intelligence. 

%GDL is a lightweight and computer-processable language to describe deterministic games with a finite number of states, players, and actions. 
%Some important game components we can specify with GDL include the initial state, legal actions of the players, state evolutions, terminal states, and the score of the players. 
However, GDL can also model ill-defined games. For instance, a game may end up in a state where some players have no legal moves. In other cases, a game may have no sequence of joint actions that allows a player to win, resulting in an unfair game. Additionally, some games may run indefinitely, making it impossible to determine a winner.

To ensure that a GDL description is usable for the GGP competitions, the notion of \textit{well-formed} descriptions has been proposed \cite{genesereth2014general}. A well-formed game should ensure that each player has at least one legal action in any non-terminal state, the game must terminate after finitely many steps, and for each player, there is a sequence of joint moves leading to one of its winning states.

Writing correct GDL descriptions can be challenging for non-experts. Automated theorem proving has been proposed to assist game design by verifying whether a GDL game satisfies desired properties~\cite{schiffel2009automated,ruan2009verification}. But theorem provers can only verify GDL descriptions, not automatically repair them, which as of now can only be done manually.

Motivated by this, and in line with recent interest and work on automatically repairing planning domain descriptions~\cite{gragera2023planning}, in this paper, we introduce the more general problem of \textit{automatically repairing game descriptions\/} given in GDL that do not comply with any given formal requirements, such as well-formedness. We consider a repair as a set of modifications to the \emph{legality} and \emph{game evolution} rules. To avoid ``redesigning'' the original game, we focus on minimal repairs. We consider game properties expressed in Game Temporal Logic (GTL)~\cite{thielscher2010temporal}, a logic similar to the Linear Temporal Logic on finite traces (LTLf)~\cite{bansal2023model} 
with only the temporal operator ``weak next'', supporting both positive and negative properties that a game should, or should not, satisfy.

Our contribution is manifold. We provide a formal definition of game repair along with theoretical results on the minimal repair problem: sufficient conditions on when certain repair problems have or do not have solutions, and tight complexity results for different computational problems related to minimal repairs. We provide an encoding based on the Answer Set Programming (ASP) technique \emph{Guess and Check}~\cite{eiter2006towards} to solve the minimal repair problem, thus offering the first automated method for repairing GDL descriptions.% and we present experimental results with a simple case study.
\subsubsection{Related work.}
 Our work is related to ASP-based approaches to formal model repair in various contexts, e.g.\ biological networks \cite{gebser2010repair}, service-based processes \cite{friedrich2010exception,lemos2019repairing}, \yifanm{logic programs~\cite{merhej2017repairing}}, and Petri nets \cite{ChiarielloIT24}. 
 
Another piece of related work is repairing unsolvable planning domain descriptions (PDDL)~\cite{gragera2023planning}. While PDDL and GDL differ, our work is more general: we can not only repair reachability to the goal (aka.\ winnability in GDL) but \emph{any\/} game properties expressible in GTL.

Given their syntactic and semantic similarity, ASP has been widely used for reasoning about GDL games, including for solving single or multiplayer games with ASP~\cite{thielscher2009answer,he2024solving}, and automatically verifying game-specific GTL properties with ASP~\cite{thielscher2010temporal}.

 %on repairing formal models. One of the contribution that is closest is repairing non-solvable descriptions provided in the Planning Domain Definition Language~\cite{gragera2023planning}. In contrast, our work not only incorporates a multi-agent setting but also addresses more generalized requirements, expressed with both positive and negative properties in GTL.

% Although not an ASP-based approach, another contribution close to our work is repairing non-solvable descriptions provided in the Planning Domain Definition Language. 

% Although that piece of
% \cite{gragera2023planning}, which investigates the repair of non-solvable descriptions provided . In contrast, our work not only incorporates a multi-agent setting but also addresses more generalized requirements, expressed with both positive and negative properties in GTL.

\subsubsection{Outline.} 
After providing necessary background on GDL and GTL, we define the GDL minimal repair problem and provide theoretical results in Section~\ref{sec:gdlrepair}. In Section~\ref{sec:encoding}, we present an ASP encoding to find minimal repairs. We demonstrate its use with a case study in Section~\ref{sec:exp}. We conclude in Section~\ref{sec:conclusion}. Proofs are available in the Appendix. % \munyque{Shall we upload the paper with the appendix to Arxiv and link it here?} \yifanm{Yes.}

%Some proofs had to be omitted; these are all available in the supplementary material. 

%% file: preliminary.tex
\section{Preliminaries}
\label{sec:prelim}
We assume readers to be familiar with basic concepts of logic programming with negation~\cite{lloyd:founda} and Answer Set Programming (ASP)~\cite{gebser2012answer}. %, \munyquem{ and refer to ~\cite{gebser2012answer} for its formal definition}. 
% I was told that we shouldn't put citations as a grammar component
%\munyque{Perhaps we could present some base definitions in the supplementary material (that is, a short description of its syntax and semantics). }

\subsection{Game Description Language}
The Game Description Language (GDL) can be used to describe the rules of any finite game with concurrent moves. GDL uses a normal logic program syntax along with the following preserved keywords used to describe the different elements of a game~\cite{genesereth2014general}\/: 
\begin{center}
    \footnotesize
    \begin{tabular}{cc}
    \toprule
    $\gdl{role}(P)$  &\!\!\!\! $P$ is a player \\
    $\gdl{base}(F)$ &\!\!\!\! $F$ is a base proposition for game positions \\ 
    $\gdl{input}(P,A)$ &\!\!\!\! Action $A$ is in the move domain of player $P$\\
    $\gdl{init}(F)$ &\!\!\!\! base proposition $F$ holds in the initial position \\
    $\gdl{true}(F)$ &\!\!\!\! base proposition $F$ holds in the current position \\
    $\gdl{legal}(P,M)$ &\!\!\!\! $P$ can do move $M$ in the current position   \\
    $\gdl{does}(P,M)$ &\!\!\!\! player $P$ does move $M$\\
    $\gdl{next}(F)$ &\!\!\!\! $F$ holds in the next position\\
    $\gdl{terminal}$ &\!\!\!\! the current position is terminal\\
    $\gdl{goal}(P,N)$ &\!\!\!\! $P$ gets $N$ points in the current position\\
    \bottomrule
    \end{tabular}
\end{center}
There are further restrictions for a set of GDL rules to be \textbf{valid}~\cite{genesereth2005general}\/: $\gdl{role}$ can appear only in facts; $\gdl{init}$ and $\gdl{next}$ can only appear as heads of rules; $\gdl{true}$ and $\gdl{does}$ can only appear in rule bodies. Moreover, $\gdl{init}$ cannot depend on $\gdl{true}$, $\gdl{does}$, $\gdl{legal}$, $\gdl{next}$, $\gdl{terminal}$, or $\gdl{goal}$ while $\gdl{legal}$, $\gdl{terminal}$, and $\gdl{goal}$ cannot depend on $\gdl{does}$. Finally, valid game descriptions must be \emph{stratified\/} and \emph{allowed\/}---such normal logic programs always admit a \textbf{finite grounding} and a unique stable model/answer set~\cite{lloyd:founda,gebser2012answer}. A valid description of a very simple game with a single player is given in Fig.~\ref{fig:gdl}.

Henceforth, we abbreviate ``GDL Description'' as $\gd$. 

A valid $\gd$ $G$ over ground terms~$\Sigma$ can be interpreted as a multi-agent state transition system\/: Let $\beta=\{f\in\Sigma\, |\, G \models base(f)\}$ be the \textbf{base propositions} and $\gamma=\{(p,a)\in\Sigma\times\Sigma\, |\, G \models input(p,a)\}$ the \textbf{move domain} for the players. Suppose that $S=\{f_{1},\myDots,f_{n}\} \subseteq \beta$ is any given position and $A=\{p_{1},\myDots,p_{k}\} \rightarrow \Sigma$ any function that assigns to each of $k\geq 1$ players an action from their move domain. In order to use the game rules~$G$ to determine the state update, $S$~needs to be encoded as a set of facts using keyword $\gdl{true}$\/: $S^{true}=\{true(f_{1}).,\myDots,~true(f_{n}).\}$ and the joint action~$A$ by a set of facts using keyword $\gdl{does}$\/: $A^{does}=\{does(p_{1},A(p_{1})).,~...,~does(p_{k},A(p_{k})).\}$. %The state transition system is then obtained as follows.

%\munyque{ For readability, it is better if we use itemize instead of bullet. I changed this through section 2 (but commented the previous version if you want to reverse it). I know it uses more space, but if equations and definitions are too close it is harder to read. If we need space, I suggest to move some proofs to the supplementary material. }
\begin{definition}[\citeauthor{schiffel2010multiagent} \citeyear{schiffel2010multiagent}]
\label{def:gdl}
The \emph{semantics\/} of a valid GDL description~$G$ is the state transition system
\begin{itemize}
\item $R=\{p \in\Sigma \,|\, G \models\gdl{role}(p)\}$ (player names) %is  the set of players;

\item $\Sinit=\{f\in\beta\,|\,G \models\gdl{init}(f)\}$ (initial state) % is the initial position;

\item $T=\{S\subseteq\beta\,|\,G \cup S^{true}\models\gdl{terminal}\}$  (terminal states)% are the terminal positions;

\item $l=\{(p, a, S)\,|\, G \cup S^{true}\models\gdl{legal}(p,a)\}$ (legal moves) %  is the legality relation, where p $\in$ R, a $\in$ $\Sigma$, and S $\in$ $2^{\Sigma}$;

\item $u(A, S)\! =\! \{f \in \beta\,|\, G \cup S^{true} \cup A^{does}\! \models\! \gdl{next}(f)\}\!$ (update) % for all A : (R $\rightarrow$ $\Sigma$) and S $\in$ $2^{\Sigma}$ is the update function;

\item $g = \{(p, v, S) \,|\, G \cup S^{true} \models \gdl{goal}(p,v)\}$ (goal value)
%is the goal relation, where p $\in$ R, v $\in$ $\mathbb{N}$, and S $\in$ $2^{\Sigma}$.
\end{itemize}
%for all finite subsets $S \subseteq \Sigma$ and functions $A:R \rightarrow \Sigma$.
%$\bullet$ R=\{p $\in$ $\Sigma$ $|$ G $\models$ $\gdl{role}$(p)\} (the players);

%$\bullet$ $S_{0}$=\{f$\in$$\Sigma$$|$ G$\models$$\gdl{init}$(f)\} (the initial position);

%$\bullet$ $T$=\{S$\in$$2^{\Sigma}$ $|$ G$\cup$$S^{true}$$\models$$\gdl{terminal}$\} (terminal positions);

%$\bullet$ l=\{(p, a, S) $|$ G $\cup$ $S^{true}$ $\models$ $\gdl{legal}$(p,a)\}, where p $\in$ R, a $\in$ $\Sigma$, and S $\in$ $2^{\Sigma}$ (the legality relation);

%$\bullet$ u(A, S) = \{f $\in$ $\Sigma$ $|$ G $\cup$ $S^{true}$ $\cup$ $A^{does}$ $\models$ $\gdl{next}$(f)\}, for all A : (R $\rightarrow$ $\Sigma$) and S $\in$ $2^{\Sigma}$ (the update function);

%$\bullet$ g = \{(p, v, S) $|$ G $\cup$ $S^{true}$ $\models$ $\gdl{goal}$(p,v)\}, where p $\in$ R, v $\in$ $\mathbb{N}$, and S $\in$ $2^{\Sigma}$ (the goal relation).
\end{definition}

Let $\gamma(p)=\{a \mid (p,a) \in \gamma\}$  be the \textbf{\em move domain of $p$}, and $B\!=\!\{S \subseteq \beta \,|\,\exists p\in R.\, \forall a \in \gamma(p).\, G \cup S^{true} \not\models legal(p,a)\}$ be all states of $G$ in which %where
some player has no legal action.
We represent a \emph{valid sequence\/} of $n$ steps starting at the initial state $S_{0}$ as $S_{0}~\xrightarrow[]{A_{0}}~S_{1}~\xrightarrow[]{A_{1}}~...~\xrightarrow[]{A_{n-1}}~S_{n}$
%We write $S_{i}$~\raisebox{-2.5pt}{$\stackrel{A_i}{\longrightarrow}$}~%~\xrightarrow[]{A_{i}} $S_{i+1}$
where $S_{i}\notin T$ for $i<n$ and all moves are legal in the corresponding state, i.e.\ $(p, A(p), S_{i}) \in l$ for each $p\in R$. %\munyque{Why do you need to say that $S_i$ is not in B? Isn't the constraint $(p, A(p), S_{i}) \in l$ enough?} % I fixed it
Valid sequences are sometimes abbreviated as $(S_{0},S_{1},\myDots,S_{n})$. A sequence \emph{terminates in $n$ steps\/} if $S_{n} \in T$.  E.g.\ in the game in Fig.~\ref{fig:gdl}, the only terminating sequence is $\{\}~$\raisebox{-3pt}{$\stackrel{\yifanm{(p,l)}}{\rightarrow}$}~$\{\gdl{loss}\}$.
A sequence \emph{ends in a non-playable state after $n$ steps} if $S_{n} \in B \setminus T$. 
% MM: Previously it was: $S_{n} \notin T \wedge S_{n} \in B$. 

%\munyque{Since $S_n \not \in T$ and $S_{n} \in B$ are not logical formulas, I would not use the logical connective $\wedge$ and write it either textually or use set operators (basically:  $A \lor B $ becomes $  A \cup B$, $ A \land B $ becomes $ A\cap B$, and $ A \land \neg B $ becomes $ A \setminus B$ ). I changed these occurrences and marked them with the colored macro.  }  % I fixed this
%\citeauthor{thielscher:AIJ12}~(\citeyear{thielscher:AIJ12}) define 
A sequence $(S_{0},\myDots,S_{m})$ is \textbf{{\em n}-max} if $m = n$ or $S_{m} \in T \cup  B$ with $m<n$~\cite{thielscher:AIJ12}. %$S_{m} \in T \vee S_{m} \in B$.
The \emph{\bf horizon} of a game is the smallest $n$ such that all $n$-max sequences terminate, end in a non-playable state, or enter a repeated state of the sequence. 

\citeauthor{genesereth2014general}~(\citeyear{genesereth2014general}) define a valid $\gd$ to be \emph{well-formed\/} if\/:
\begin{enumerate}
    \item For each $p$, there is a terminating sequence $(\Sinit,\myDots,S_{m})$ with $G \cup S_{m}^{true} \models goal(p,100)$ (\textbf{weak winnability}).
    \item No sequence ends in a non-playable state (\textbf{playability}).
    \item All play sequences terminate (\textbf{termination}). 
\end{enumerate}
%GDLs
Our $\gd$ in Fig.~\ref{fig:gdl} is not well-formed because it is not weakly winnable since action $r$ is not legal for the player. A standard requirement for $\gd$'s used in the GGP competition is to be well-formed~%\yifanm
{\cite{genesereth2014general}}.
%\yifan{I changed the reference to the 2014 book}
%\michael{I've added a Stanford link, but this may not work forever. Yifan, please check if Genesereth\&Thielscher 2014 has all you need.}
% \munyque{I no longer can access the technical report from \cite{genesereth2005general}. Should we cite \cite{GeneserethLP05} instead (or additionally)?} \yifan{I downloaded a PDF version 2 years ago, but I can't find the online version either. Can we put the PDF in the supplementary material? Because things like depends on, predicates dependency graph, no infinite recursion restriction are not in the 2005 version} \munyque{It would be weird to add a paper in the supplementary material. Perhaps it is better to keep the citation as it is and contact the authors, saying we can no longer find the paper online (I imagine Michael knows the authors, so maybe we can ask him). }

We define a GD to be \textbf{{\em n}-well-formed} if it is well-formed and the game has a horizon of \emph{no more than} $n$.
\subsection{The Game Temporal Logic}
%\munyque{We need to clarify the connection with GTL in the first paragraph of this section. See suggestion:} % fixed

The Game Temporal Logic (GTL) \cite{thielscher2010temporal} is defined over $\gd$s and allows the formulation of properties that involve finitely many successive game states. 
\begin{definition}
     The set of \emph{GTL formulas\/} over a $\gd$ $G$ is\/:
\[
  	\varphi  ::= q \mid \varphi \land \varphi \mid \neg \varphi \mid \X \phi
\]
where $q$ is an atom $\textbf{true(f)}$ for some $f \in \beta$, or $\textbf{legal(p,a)}$ for some $(p,a) \in \gamma$, or any other ground atom of predicates of $G$ but which is neither $\textbf{init}$ nor $\textbf{next}$ and which does not depend, in $G$, on \textbf{does}. %\munyque{What does "depend" mean? That $q$ is not in the body of a rule with "does"?}. 
$\X$ is the temporal operator ``next''. Standard connectives like $\vee$ and $\rightarrow$ are defined as usual. 

The \emph{degree\/} $deg(\phi)$ of a formula $\phi$ is the maximal ``nesting'' of the unary operator $\X$ in $\phi$.
\end{definition}
\begin{definition} %[GTL semantics \cite{thielscher2010temporal}]
    Let $G$ be a valid $\gd$ and $\phi$ a GTL formula with $deg(\phi)=n$. We say \emph{$G$ satisfies $\phi$} (written $G \models_{t} \phi$) iff for all $n$-max sequences $(S_{0},\myDots,S_{m})$ we have $G,(S_{0},\myDots,S_{m}) \models_{t} \phi$ as the following inductive definition\/:
    %\yifanm{(we write $G\not\models_{t} \phi$ if $G\models_{t}\phi$ doesn't hold)}:
    % Let $G$ be a valid $\gd$, $S_{0}$ a state, and $\phi$ a GTL formula with $deg(\phi)=n$. We say \emph{$G$ at $S_{0}$ satisfies $\phi$} (written $G,S_{0} \models_{t} \phi$) iff for all $n$-max sequences $(S_{0},\myDots,S_{m})$ we have  $G,(S_{0},\myDots,S_{m}) \models_{t} \phi$ as per the following inductive definition\/:
    % %where $i=0,\ldots,m$: \michael{what is $t$? what is $i$?} \yifan{Fixed}
%   \munyque{In \cite{thielscher2010temporal}, G is not included before the sequences. They write simply $(S_0, \dots S_m) \models_t \varphi$. Any reason to change this?} \yifan{the reason is that we are doing GDL repair and G can be changed} MM: Ok
    % \begin{align*}
    % G,(S_{i},\myDots,S_{m}) &\models_{t} q &\text{ iff }& G \cup S_{i}^{true} \models q %MM: we can remove this: (q \text{ ground atom})
    % \\
    % G,(S_{i},\myDots,S_{m}) &\models_{t} \neg \phi &\text{ iff }& G, (S_{i},\myDots,S_{m})\not\models_{t} \phi\\
    % G,(S_{i},\myDots,S_{m}) &\models_{t} \phi_{1} \wedge \phi_{2} &\text{ iff }& G, (S_{i},\myDots,S_{m})\models_{t} \phi_{1} \\
    % & & & \text{ and } G, (S_{i},\myDots,S_{m})\models_{t} \phi_{2}\\
    % G,(S_{i},\myDots,S_{m})&\models_{t} \X \phi &\text{ iff }& G,(S_{i+1},\myDots,S_{m})\models_{t} \phi \, (i<m) \\
    % G,(S_{m}) & \models_{t} \X \phi  & \text{iff} & \text{~~True} 
    % \end{align*}
    \begin{align*}
    G,(S_{i},\myDots,S_{m}) &\models_{t} q \hspace{4.3em} \text{ iff } G \cup S_{i}^{true} \models q %MM: we can remove this: (q \text{ ground atom})
    \\
    G,(S_{i},\myDots,S_{m}) &\models_{t} \neg \phi \hspace{3.5em} \text{ iff } G, (S_{i},\myDots,S_{m})\not\models_{t} \phi \\
    G,(S_{i},\myDots,S_{m}) &\models_{t} \phi_{1} \wedge \phi_{2} \hspace{1.5em} \text{ iff } G, (S_{i},\myDots,S_{m})\models_{t} \phi_{1} \\
    &  \hspace{6.5em} \text{ and } G, (S_{i},\myDots,S_{m})\models_{t} \phi_{2}\\
    G,(S_{i},\myDots,S_{m})&\models_{t} \X \phi \text{ iff } G,(S_{i+1},\myDots,S_{m})\models_{t} \phi \, (i<m) \\
    G,(S_{m}) & \models_{t} \X \phi  \hspace{3.4em}  \text{always} %\text{iff}  \text{~~True} 
    \end{align*}
\end{definition}
\iffalse
\begin{definition}[GTL semantics \cite{thielscher2010temporal}]
    Let $G$ be a valid description, $S_{0}$ be a state, and $\phi$ be a GTL formula with $deg(\phi)=n$. We say $G,S_{0} \models_{t} \phi$ iff for all $n$-max sequences $(S_{0},\myDots,S_{m})$ we have $(S_{0},\myDots,S_{m}) \models_{t} \phi$ according to the following definition:\\
    $G,(S_{i},\myDots,S_{m}) \models_{t} q$~~~~iff~~~$G \cup S_{i}^{true} \models q$ \\
    $G,(S_{i},\myDots,S_{m}) \models_{t} \neg \phi$~~~~iff~~~$G, (S_{i},\myDots,S_{m})\not\models_{t} \phi$\\
    $G,(S_{i},\myDots,S_{m}) \models_{t} \phi_{1} \wedge \phi_{2}$ iff $G, (S_{i},\myDots,S_{m})\models_{t} \phi_{1}$ \\
    \hspace*{11.5em} and $G, (S_{i},\myDots,S_{m})\models_{t} \phi_{2}$\\
    $G,(S_{i},\myDots,S_{m})\models_{t} \X \phi$~iff~$G,(S_{i+1},\myDots,S_{m})\models_{t} \phi$ (i$<$m) \\
    \hspace*{2.9em} $G,(S_{m}) \models_{t} \X \phi$
\end{definition}
\fi
%\textbf{A crucial part} here is $G,(S_{m}) \models_{t} \X \phi$: If we reach the end of a sequence, every formula of the form $\X \phi$ must hold. \munyque{This sentence can be removed if we need space - and it should be clear from the semantics (we can also write "$G,(S_{m}) \models_{t} \X \phi$ iff true" in the semantics)}

% The GTL model checking task decides if $G,\Sinit \models_{t} \phi$ (\yifanm{?abbreviated as $G \models_{t} \phi$}) holds, where $\phi$ is a GTL formula.

The \emph{GTL model checking task} decides if $G\models_{t} \phi$ holds, where $\phi$ is a GTL formula. Due to its syntactic and semantic similarities to GDL, ASP is a natural choice for GTL model checking. The encoding consists of (1)~an ASP encoding of the GTL formula, (2)~an ASP representation of the game rules, and (3)~an action generator in ASP. \yifanm{Detailed description and examples of the encoding are available in the original article~\cite{thielscher:AIJ12}.}
\begin{definition}[ASP encoding of GTL~\cite{thielscher:AIJ12}] \label{def:aspgtl}
Let $\phi$ be a GTL, $i \in \mathbb{N}$, and $\eta(\phi,i)$ a function that gives a unique atom of arity 0 for every $\phi$ and~$i$. The \emph{encoding $\enc(\phi,i)$ of $\phi$ at level $i$} %denoted $Enc(\phi,i)$, 
is recursively defined (below, $q(\vec{t\,})$ means the predicate $q$ with arguments $\vec{t}$)\/:
%\michael{use $\vec t$ instead of $\vv t$ everywhere} \yifan{fixed}
%\munyque{What is $\vv{t}$?} 
\begin{itemize}
    \item $\enc(q(\vec{t\,}),i) = \{\eta(q(\vec{t\,}),i) \texttt{:-}\, q(\vec{t},i).\}$ 
    \item $\enc(\neg \phi,i) = \{\eta(\neg \phi,i) \texttt{:-}\,~not~\eta(\phi,i).\} \cup \enc(\phi,i)$
    \item $\enc(\phi_{1} \wedge \phi_{2},i) = \{\eta(\phi_{1} \wedge \phi_{2},i) \texttt{:-}\, \eta(\phi_{1},i),\eta(\phi_{2},i).\} \\ \hspace*{7.9em} \cup \enc(\phi_{1},i) \cup \enc(\phi_{2},i)$ 
    \item $\enc(\X \phi,i)=\{\eta(\X \phi,i) \texttt{:-} terminal(i)., \\ \hspace*{4.9em} \eta(\X \phi,i) \texttt{:-} no\_play(i).,\\     \hspace*{4.9em} \eta(\X \phi,i) \texttt{:-} \eta(\phi,i + 1).\} \cup \enc(\phi,i + 1)$
\end{itemize}
%$ $\\
%$\bullet$ $Enc(q(\vv{t}),i) = \{\eta(q(\vv{t}),i) \texttt{:-} q(\vv{t},i).\}$\\
%    $\bullet$ $Enc(\neg \phi,i) = \{\eta(\neg \phi,i) \texttt{:-} \neg \eta(\phi,i).\} \cup Enc(\phi,i)$\\
%    $\bullet$ $Enc(\phi_{1} \wedge \phi_{2},i) = \{\eta(\phi_{1} \wedge \phi_{2},i) \texttt{:-} \eta(\phi_{1},i),\eta(\phi_{2},i).\} \cup \\ \hspace*{4.9em} Enc(\phi_{1},i) \cup Enc(\phi_{2},i)$ \\
%    $\bullet$ $Enc(\X \phi,i)=\{\eta(\X \phi,i) \texttt{:-} terminal(i)., \\ \hspace*{4.9em} \eta(\X \phi,i) \texttt{:-} no\_play(i).,\\     \hspace*{4.9em} \eta(\X \phi,i) \texttt{:-} \eta(\phi,i + 1).\} \cup Enc(\phi,i + 1)$
\end{definition}
\begin{definition}[\citeauthor{thielscher2010temporal} \citeyear{thielscher2010temporal}] \label{def:temporal}
The \textbf{Temporal-Extension} with horizon $n\geq 0$ of a valid $\gd$ $G$ (denoted $\pext^{n}(G)$) is defined as $\pext^{n}(G)=\bigcup_{0 \leq i \leq n} \{c^{i}\mid c \in G\}$ where $\cdot^{i}$ replaces each occurrence of

\begin{itemize}
    \item  $\gdl{init}(f)$ by $\gdl{true}(f,0)$; and $\gdl{next}(f)$ by $\gdl{true}(f,i+1)$.
    
    \item $\gdl{q}(\vec{t}\,)$ by $\gdl{q}(\vec{t},i)$ if $q \notin \{init,next\}$ is a predicate symbol that depends on $\gdl{true}$ or $\gdl{does}$.
\end{itemize}
%\michael{when $\vec t$ is followed by a closing parenthesis, always add some space, as in: $(\vec t\,)$}
  %  $\bullet$ $\gdl{init}(f)$ by $\gdl{true}(f,0)$; and $\gdl{next}(f)$ by $\gdl{true}(f,i+1)$.
    
  %  $\bullet$ $\gdl{q}$$(\vv{t})$ by $\gdl{q}(\vv{t},i)$ if $q \notin \{init,next\}$ is a predicate symbol that connects to $\gdl{true}$ or $\gdl{does}$.
\end{definition}
 
\begin{definition}[\citeauthor{thielscher:AIJ12} \citeyear{thielscher:AIJ12}]  \label{def:aspgenerator}
The \textbf{action generator} requires each player to make a legal action at each playable, non-terminal state up to step $n$\/: $\plegal^{n}$ consists of the following clauses $P_{i}$ for each $0 \leq i \leq n$.%\\
   % $\bullet$ $no\_play(i)\texttt{:-}\,role(R),not~legal(R,A,i):input(R,A).$\\
   % $\bullet$ $end(i) \texttt{:-}\, terminal(i).$\\
   % $\bullet$ $end(i)\texttt{:-}\,~no\_play(i).$\\
   % $\bullet$ $end(i)\texttt{:-}\, end(i-1).$ for each $i > 0$.\\
   % $\bullet$ $1\{does(R,A,i):input(R,A)\}1 \texttt{:-}\, not~end(i), role(R).$\\
   % $\bullet$ $\texttt{:-}\,~does (R, A , i), not~legal(R, A , i).$
   \begin{itemize}
       \item $no\_play(i)\texttt{:-}\,role(R),not~legal(R,A,i):\!input(R,A).$
       \item $end(i) \texttt{:-}\,~~1~\{terminal(i);~no\_play(i)\}.$
       %\item  $end(i)\texttt{:-}\,~no\_play(i).$
       \item  $end(i)\,\texttt{:-}\  end(i-1).$ for each $i > 0$.
       \item $1\{does(R,A,i):input(R,A)\}1 \texttt{:-}\, not~end(i), role(R).$
       \item  $\texttt{:-}\,~does (R, A , i), not~legal(R, A , i).$
   \end{itemize}
\end{definition}

%\munyque{Some suggestions IF we need to save space: 
     %\item We could omit the statement of theorem 1 (and when it is needed we can write theorem X of \citeauthor{thielscher:AIJ12}) 
%     We could give short versions of Definitions 4, 5, and 6 by presenting their most relevant part and having the full definition in the appendix. For instance, for Definition 4 we could omit $\neg \varphi$ and  $\varphi \land \varphi$, because those are quite standard.} 
% \yifan{We can eliminate some cases in Definition 4-6, but Theorem 1 is a summary of half of the theoretical results of the 2010/12 theorem proving conference/journal paper. We need these notations in Corollary 1 and the disjunctive logic program construction.}

% Verifying $G,\Sinit \models_{t} \phi$ is then achieved by checking that there is no $deg(\phi)$-max with $G,(\Sinit,\myDots,S_{m}) \not\models_{t} \phi$.
Verifying $G \models_{t} \phi$ is then achieved by checking that there is no $deg(\phi)$-max with $G,(\Sinit,\myDots,S_{m}) \not\models_{t} \phi$.
\begin{theorem}[\citeauthor{thielscher:AIJ12} \citeyear{thielscher:AIJ12}] 
\label{theorem:basecase}
Let $G$ be a valid $\gd$, and $\phi$ be a GTL formula with $deg(\phi)$=$n$. Then, $G \models_{t} \phi$ iff the program $\plegal^{n} \cup \pext^{n}(G) \cup \enc(\phi,0) \cup \{\texttt{:-}\, \eta(\phi,0)\}$ has no stable model.
\end{theorem}
\yifanm{We refer to the size of a logic program as the size of the
ground program.} 
Since the size of the ASP program according to Definitions~\ref{def:aspgtl}--\ref{def:aspgenerator} is polynomial w.r.t. the size of $G$ and $\phi$, it is known that GTL model-checking is in co-NP.
%\michael{What's confusing is the use of $S_0$ as initial state and as any state from which a sequence starts. Should use $\Sinit$ instead of $S_0$} 

%% file: problem-definition.tex
\section{The GDL Repair Problem}
\label{sec:gdlrepair}
%\munyquem{Next, we introduce the problem of repairing GDL games.} \munyque{(if space)}
In this section, we propose a formal definition of GDL game repair and then present theoretical results on this problem.

\subsection{Problem Definition}
\label{sec:problem-def}
%\munyque{For me, a "game description" should be denoted with only GD (if we want to abbreviate), because the L in GDL denotes language.}
% I fixed it
We first define some auxiliary notation. For a valid $\gd$ $G$ with base propositions~$\beta$ and move domain~$\gamma$ (cf. Section~\ref{sec:prelim}):
    \begin{itemize}
        \item $\mathcal{N}=\{next(f) \mid f \in \beta\}$ (the domain of $\gdl{next}$)
        \item $\mathcal{L}=\{legal(p,a) \mid (p,a) \in \gamma\}$ (the domain of $\gdl{legal}$)
        \item $\mathcal{F}=\{true(f),\,not~true(f) \mid f \in \beta\}$
        \item $\mathcal{A}=\{does(p,a),\,not~does(p,a) \mid (p,a) \in \gamma\}$
        \item $|P|=$ the total number of rules of a grounded program $P$.
        \item If $r$ is a grounded rule, then $hd(r)$ denotes the atom in the head of $r$, and $bd(r)$ the set of literals in the body of $r$.
    \end{itemize}
For a GTL formula $\phi$ and $op\in\{\wedge,\vee\}$ we define the abbreviation $nest(\phi,op,n)$ recursively as follows: 

% I changed back to the bullets here because they look quite simple
% \begin{itemize}
%     \item $nest(\phi,\wedge,0)=nest(\phi,\vee,0)=\phi$
%     \item $nest(\phi,\wedge,n)=\phi~\wedge~(\X~nest(\phi,\wedge,n-1))$
%     \item $nest(\phi,\vee,n)=\phi~\vee~(\X~nest(\phi,\vee,n-1))$.
% \end{itemize}
\begin{itemize}
    \item $nest(\phi,op,0)=\phi$
    \item $nest(\phi,op,n)=\phi~op~(\X~nest(\phi,op,n-1))$.
\end{itemize}
%$\bullet$ $nest(\phi,\vee,n)=\phi~\vee~(\X~nest(\phi,\vee,n-1))$. \\
%%MM: removed: Note that $nest$ is \textbf{not} a GTL operator. It is introduced only for abbreviation purposes. 
For example, $nest(\gdl{true}(\gdl{win}),\vee,2)$ denotes the GTL formula $\gdl{true}(\gdl{win}) \vee (\X (\gdl{true}(\gdl{win}) \vee (\X \gdl{true}(\gdl{win}))))$.

    % $\bullet$ $\mathcal{N}=\{next(f)~|~f \in \beta\}$
    
    % $\bullet$ $\mathcal{L}=\{legal(p,a)~|~(p,a) \in \gamma\}$
    
    % $\bullet$ $\mathcal{F}=\{true(f),not~true(f)~|~f \in \beta\}$%$ \cup \{|G \models base(f)\}$,
    
    % $\bullet$ $\mathcal{A}=\{does(p,a),not~does(p,a)~|~(p,a) \in \gamma\}$ %\cup \{not~does(p,a)|G \models input(p,a)\}$
    
    % $\bullet$ $Dom= (\{c\}\times (\{\emptyset\} \cup \mathcal{L} \cup \mathcal{N}))\cup (\{+,-\}\times (\mathcal{A} \cup \mathcal{F}))$.
    
    % $\bullet$ $|P|$: the total number of rules of a grounded program $P$.
    
    % $\bullet$ If $r$ is a grounded rule, then, $hd(r)$: the atom of the head of $r$, and $bd(r)$: the set of literals in the body of $r$.
%For the final item, consider the rule
For the definition of repair, we assume that the $\gd$ have been \textbf{grounded} and transformed into a simplified form according to the following definition.
\begin{definition} Suppose $G$ is a valid, grounded $\gd$.
%A rule $r \in G$ is called a legal (resp. next) rule if $hd(r)$ is an instance of $\gdl{legal}$ (resp. $\gdl{next}$).
$G$ is in \emph{\bf restricted form} iff for every rule $r\in G$ the following holds\/:
    \begin{itemize}
        \item ``$\gdl{legal}$'' does not appear in $bd(r)$.
        \item ``$\gdl{does}$'' does not appear in $bd(r)$ unless $hd(r)\in\mathcal{N}$.    
        \item If $r$ is a \emph{\bf legal rule}, i.e.\ $hd(r) \in \mathcal{L}$, then $bd(r) \subseteq \mathcal{F}$.
        \item If $r$ is a \emph{\bf next rule}, i.e.\ $hd(r)\in\mathcal{N}$, then $bd(r) \subseteq \mathcal{A} \cup \mathcal{F}$.
    \end{itemize}    
    % $\bullet$ $\gdl{legal}$ does not appear in $bd(r)$
    
    % $\bullet$ $\gdl{does}$ cannot appear in $bd(r)$ unless $r$ is a next rule.    
    
    % $\bullet$ If $r$ is a legal rule, $bd(r) \subseteq \mathcal{F}$
    
    % $\bullet$ If $r$ is a next rule, $bd(r) \subseteq \mathcal{A} \cup \mathcal{F}$
\end{definition}
%\munyquem{, e.g., the body of a legal rule contain statements about the truth or falsehood of base propositions. }
%\munyque{For the last two items, maybe we can add the description in words? Example above. } \yifan{fixed below}
The last two items state that for a $\gd$ in restricted form, the body of any legal (resp.\ next) rule can only have positive or negative atoms of the predicate $\gdl{true}$ (resp. $\gdl{true}$ or $\gdl{does}$).

\michaelm{Any valid $\gd$ has a finite grounding, and GDL in restricted form is as expressive as GDL as it can still express all finite perfect information games based on the same construction as by~\citeauthor{thielscher2011general}~(\citeyear{thielscher2011general}) for GDL-II. Thus, for simplicity, we only consider $\gd$ in restricted form.} Informally speaking, we are given a GD that satisfies (violates) undesired (desired, resp.) properties formulated in GTL. Then, to \emph{repair\/} a GD, we allow the following possible changes\/:
\begin{enumerate}
     \item modifying existing legal/next rules by      \begin{itemize}
         \item deleting the entire rule or some literals from its body,
         \item changing the head of a legal rule to a new atom in $\mathcal{L}$ (or the head of a next rule to a new atom in $\mathcal{N}$, resp.), or
         \item adding one or more literals in $\mathcal{F}$ to the body of a legal rule (or literals in $\mathcal{A} \cup \mathcal{F}$ to the body of a next rule);
     \end{itemize}
     \item adding a \emph{bounded\/} number of new legal/next rules.
\end{enumerate}
The goal is to make minimal changes so that the resulting $\gd$ \emph{ \bf satisfies or dissatisfies} some GTL properties. We only allow modifications to legal/next rules because changing other rules, such as the definition of base propositions, initial state, and goal, would be akin to defining a new game rather than repairing an existing one. %\yifan{TODO: explain the bound $|G_{E}|$}
%\munyque{Isn't it enough to have each set $G_L$, ..., $G_E$ to be finite?}
To ensure that the set of possible repairs is always finite, we also assume a given bound on the number of legal/next rules that can be added.

Since the input $\gd$ is in restricted form, we may assume w.l.o.g. that it has a form $G=G_{L} \cup G_{N} \cup G_{E} \cup G_{R}$, where
\begin{itemize}
    
    \item %$\bullet$
    $G_{L}$ (resp. $G_{N}$) contains all the legal (resp. next) rules.

    \item %$\bullet$
    $G_{E}$ has a fixed number of ``empty'' rules of the form $\emptyset$.

    %\munyque{I suggest just calling "empty body rules"} \munyque{Are the double dots intentional? (ie. part of the syntax).  If so, don't we need it in Def 9?} \yifan{I eliminate the second ``.'', and should we call it ``place holder rules'' or something similar?}

    \item $G_{R}$ has all other rules (e.g., rules defining $init,goal,\myDots$).
\end{itemize}

$G_{E}$ is the section reserved for new legal/next rules, with     
$\emptyset$ being a ``placeholder'' symbol that does not appear elsewhere in $G$.
Let $G_{O}=G_{L} \cup G_{N}$ and $G_{C}=G_{L} \cup G_{N} \cup G_{E}$. We associate each $r \in G_{C}$ with a unique ID from the set $\mathcal{I}=\{1,\myDots,|G_{C}|\}$,
%\michael{should we add $G$ to $\mathcal I$?} \yifan{It's a fixed number before and after the repair}
and we denote the $i$-th rule as $r_{i}$. We assume a simple order by which the IDs of the rules in $G_{L}$ are smaller than the IDs of rules in $G_{N}$, which in turn are smaller than the IDs of rules in $G_{E}$. %\munyque{id $\to$ ID}
        % (to be used below (cf.\ $\defshort$~\ref{def:repset}) for specifying indichanges)
A \emph{change tuple\/} formally defines an individual change to a $\gd$. Repairing a $\gd$ involves applying a set of change tuples (aka.\ \emph{repair\/}) to the $\gd$ \emph{simultaneously}. \yifanm{To conveniently describe individual changes to a $\gd$, we define the following notation:}
\begin{itemize}
    \item $Dom= (\{c\}\times (\{\emptyset\} \cup \mathcal{L} \cup \mathcal{N}))\cup (\{+,-\}\times (\mathcal{A} \cup \mathcal{F}))$
\end{itemize}
\begin{definition}[Change tuples and repairs]
\label{def:repset}
    Let $G$ be a valid game description. A \emph{\bf change tuple} on $G$ has the form $\langle i,(tp,l) \rangle\in \mathcal{I}\times\Dom$. A \emph{\bf repair} $\mathcal{R}$ is a \emph{\bf set} of change tuples. Suppose $L^{+}_{i}=\{l\,|\,\langle i,(+,l) \rangle \in \mathcal{R}\}$ and $L^{-}_{i}=\{l\,|\,\langle i,(-,l) \rangle \in \mathcal{R}\}$, for each $i \in \mathcal{I}$. The repair is \emph{\bf valid} iff for all $i \in \mathcal{I}$, all of the following hold:
    \begin{enumerate}[a)]
        \item \label{con1} $|\{h\,|\,\langle i,(c,h) \rangle \in \mathcal{R}\}| \leq 1$.
        \item \label{con2} $L_{i}^{+} \subseteq \left(\mathcal{A} \cup \mathcal{F}\right) \setminus bd(r_{i})$ and $L_{i}^{-} \subseteq bd(r_{i})$.
        \item \label{con3} If $r_{i} \in G_{L}$ and $\langle i,(c,h) \rangle \in \mathcal{R}$, then $h \in \{\emptyset\} \cup \mathcal{L}$.
        \item \label{con4} If $r_{i} \in G_{N}$ and $\langle i,(c,h) \rangle \in \mathcal{R}$, then $h \in \{\emptyset\} \cup \mathcal{N}$.
        \item \label{con5} If $hd(r_{i}) \in \mathcal{L}$ \textbf{or} $\langle i, (c, h) \rangle  \in \mathcal{R}$ for $h \in \mathcal{L}$, then $L^{+}_{i} \subseteq \mathcal{F}$.
        % there exists no $h_{1} \neq h_{2}$ such that $\langle i,(c,h_{1}) \rangle \in \mathcal{R}$ and $\langle i,(c,h_{2}) \rangle \in \mathcal{R}$.        
    \end{enumerate}
    The resulting $\gd$, written $rep(G,\mathcal{R})$, after \emph{\bf applying a valid repair} $\mathcal{R}$ to $G$ is $r_{1}'\cup \myDots \cup r_{|G_{C}|}' \cup G_{R}$ where, for each  $i \in \mathcal{I}$:
    \begin{enumerate}[i)]
        \item $hd(r_{i}')\!=\!h$ if $\langle i,(c,h) \rangle \!\in\! \mathcal{R}$, otherwise $hd(r_{i}')\!=\!hd({r_{i}})$ 
        %h$ if some $\langle i,(c,h) \rangle \in \mathcal{R}$ ;else $hd(r_{i}')=hd({r_{i}})$
        \item $bd(r_{i}')\!=\!bd(r_{i}) \cup L_{i}^{+} \setminus L_{i}^{-}$ if $hd(r'_{i}) \neq \emptyset$;\,else $bd(r_{i}')=\{\}$. %if $h_{i} \neq \emptyset$; else, $bd(r_{i}')\!=\!\{\}$. 
    \end{enumerate}
        %\item 
        
   % A repair set is \emph{\bf valid} iff it is \emph{\bf feasible}, and $rep(G,R)$ is a valid $\gd$. 
\end{definition}
Intuitively, the change tuple $\langle i,(c,h) \rangle$ means changing $hd(r_{i})$ to $h$. Deleting rule~$i$ is achieved by setting $hd(r_{i})$ to $\emptyset$ with the change tuple $\langle i,(c,\emptyset) \rangle$. $\langle i,(+,l) \rangle$ with $l \notin bd(r_{i})$ means adding $l$ to the body of $r_{i}$ while $\langle i,(-,l) \rangle$ with $l \in bd(r_{i})$ means deleting $l$ from $r_{i}$'s body. Since a GDL rule can only have one head atom, condition~\ref{con1} must hold for a valid repair. Conditions~\ref{con3} and~\ref{con4} imply that a legal rule cannot be changed to a next rule or vice versa. Condition~ii) for the repair operation ensures that a rule $r_{i}$ in the resulting $\gd$ has an empty body when it has an empty head. %$hd(r'_{i})=\emptyset$.

  Adding a new rule with head $h$ is represented by $\langle i,(c,h) \rangle$ along with zero or more $\langle i,(+,l) \rangle$ tuples for some $h \neq \emptyset$ and $r_{i} \in G_{E}$. Up to $|G_E|$ new rules can thus be added. Condition~\ref{con5} requires that the body of any legal rule in the repaired $\gd$ does not contain a literal $l \in \mathcal{A}$, because $\gdl{legal}$ cannot depend on $\gdl{does}$ in a valid $\gd$. 

Since both $|G_{C}|$ as well as the domains for the base propositions and moves are fixed, the number of repairs is finite, and the space complexity of a repair is $O(|G_{C}| \cdot (|\mathcal{F}| + |\mathcal{A}|))$. 
% Intuitively, if $\langle h_{i},L_{i}^{-},L_{i}^{+} \rangle$ is the $i$-th triple of a repair set and $h_{i} \neq hd(r_{i})$, the head of $r_{i}$ is changed to $h_{i}$ and we add (resp. delete) the set of literals in $L^{+}$ (resp. $L^{-}$) to $bd(r_{i})$. Note that, if $r_{i}$ is a legal rule, $L^{+}_{i}$, the literals we add to $bd(r_{i})$ cannot contain any $l \in \mathcal{A}$, because in a valid $\gd$, $\gdl{legal}$ cannot depend on $\gdl{does}$. If $h_{i}=\emptyset$, it is equivalent to deleting $r_{i}$. We can add at most $|G_{E}|$ new rules. Adding a new rule with head $h_{i}$ and body $L^{+}_{i}$ is simulated by specifying $h_{i} \neq \emptyset$ and $L^{+}_{i}$ for some $r_{i} \in G_{E}$.

\begin{example}
\label{example:repairset}
Let $G$ be the $\gd$ consisting of the rules in Fig.~\ref{fig:gdl} and $G_{E}=\,\{ \mbox{\tt [c4]}\,\emptyset\}$. Then, $G_{L}=c_{1}$, $G_{N}=c_{2} \cup c_{3}$, $G_{R}=c_{r}$, and $|\mathcal{I}|=|G_{C}|=4$. A valid (but not necessarily minimal) 
repair $\mathcal{R}$ of $G$ is $$\{\langle 3,(-,\gdl{does}(p,r)) \rangle,\ \langle 4,(c,legal(p,r)) \rangle\}$$ %, meaning:
    % \begin{itemize}
    %     \item 1:$\langle \gdl{legal}(p,l) ,\{\},\{\} \rangle$~~~2:$\langle \gdl{next}(loss) ,\{\},\{\} \rangle$
    %     \item 3:$\langle \gdl{next}(win) ,\{does(p,r)\},\{\} \rangle$~~~~4:$\langle \gdl{legal}(p,r) ,\{\},\{\} \rangle$
    % \end{itemize}
    % % $\bullet$ 1:$\langle \gdl{legal}(p,l) ,\{\},\{\} \rangle$~~~2:$\langle \gdl{next}(loss) ,\{does(p,l)\},\{\} \rangle$ 
    % $\bullet$ 3:$\langle \gdl{next}(win) ,\{\},\{\} \rangle$~~~~4:$\langle \gdl{legal}(p,r) ,\{\},\{\} \rangle$\\
%The two change tuples of $\mathcal{R}$ are $\langle 3,(-,\gdl{does}(p,r)) \rangle$ and $\langle 4,(c,legal(p,r)) \rangle$. 
In words, remove $\gdl{does}(p,r)$ from the body of $c_{3}$ and change the head of rule $c_4$ (the empty rule) to $\gdl{legal}(p,r)$. Thus, $rep(G,\mathcal{R})=c_{1}   \cup c_{2} \cup \{next(win).\} \cup \{legal(p,r).\} \cup c_{r}$.  

%\munyque{Is the repaired game weakly winnable? We should mention it.} \yifan{Fixed}
\end{example} 
To capture %The cost function captures 
the total cost of a repair to $G$, we consider a cost function based on the cost of each individual change\/:  %\munyque{Maybe we can stress that this is a more general setting than considering a single cost for all modifications. E.g. \cite{ChiarielloIT24} considers a bound on the number of modifications.}
\begin{definition}[cost function]
\label{def:cost}
 Let $G$ be a valid $\gd$. A \emph{cost function} over $G$ is a mapping $cost:\mathcal{I} \times Dom \rightarrow \mathbb{N}_{>0}$. The
 \emph{cost $\cost(\mathcal{R})$ of a repair} $\mathcal{R}$ to $G$ is $\sum_{\langle i,dom \rangle \in \mathcal{R}} cost(i,dom)$. %such that 
% \begin{itemize}
%     \item $cost(i,(c,hd(r_{i})))=0$ for all $1 \leq i \leq |G_{C}|$
%     %\item $cost(i,(c,\emptyset)))=0$ for all $|G_{O}| < i \leq |G_{C}|$
%     \item $cost(i,dom)>0$, for other $i \in \mathcal{I}$ and $dom \in Dom$
% \end{itemize}
% $\bullet$ $cost(i,(c,hd(r_{i})))=0$ for all $1 \leq i \leq |G_{O}|$

% $\bullet$ $cost(i,(c,\emptyset)))=0$ for all $|G_{O}| < i \leq |G_{C}|$

% $\bullet$ $cost(i,dom)>0$, for other $i \in \mathcal{I}$ and $dom \in Dom$\\
%for all change tuples $\langle i,dom \rangle \in \mathcal{R}$. 
\end{definition}
\begin{definition}[Minimal repair problem]
\label{def:rp}
A \textbf{repair task} is a tuple $\langle G,\Phi^{+},\Phi^{-},cost\rangle$, with $G$ a $\gd$, $\Phi^{+}$ and $\Phi^{-}$ sets of GTL formulas over $G$, and a cost function.

A \textbf{solution} is a valid repair $\mathcal{R}$ with $rep(G,\mathcal{R}) \models_{t} \phi^{+}$ for all $\phi^{+} \in \Phi^{+}$, and $rep(G,\mathcal{R}) \not\models_{t} \phi^{-}$ for all $\phi^{-} \in \Phi^{-}$.
%The \textbf{repair cost} is the cost of $\mathcal{R}$ w.r.t. $cost$ (Def.~\ref{def:cost}). 

The \textbf{minimal repair problem} (MRP) takes a repair task $\mathcal{T}$ as input and outputs the lowest-cost solution repair $\mathcal{R}$ to  $\mathcal{T}$.
\end{definition}
\begin{example}
\label{example:repair}
%\emph{Well-formedness} is a critical property for GDL games. 
Consider $\Phi^{+}\!=\!\{\psi_{end}(n),\psi_{play}(n)\}$ and $\Phi^{-} \!=\!\{\psi_{loss}(p,n) \mid G\models role(p)\}$ for some $\gd$ $G$ and $n>0$, then the answer to the MRP $\langle G,\Phi^{+},\Phi^{-},cost\rangle$ is the lowest-cost repair to make $G$ to be n-well-formed, where
%\textbf{well-formed}, and the repaired game has an horizon of at most $N$
\begin{itemize}
    \item $\psi_{end}(n) \ \DefMath\ \ nest(terminal,\vee,n)$
    \item $\psi_{lg}(p)=\bigvee_{a \in \gamma(p)} legal(p,a)$
    \item $\psi_{play}(n) \ \DefMath\ \ nest(terminal \vee \bigwedge_{G \models role(p)} \psi_{lg}(p),\wedge,n)$
    \item $\psi_{loss}(p,n) \ \DefMath\ \ nest(\neg terminal \vee \neg goal(p,100),\wedge,n)$
\end{itemize}
$\psi_{end}(n)$ requires all $n$-max sequences to terminate or end in a non-playable state. $\psi_{play}(n)$ requires all players to have some legal actions per non-terminal state. $\psi_{loss}(p,n) \in \Phi^{-}$ ensures some $n$-max sequence terminate with $goal(p,100)$.

For the $\gd$ in Example~\ref{example:repairset}, if $n\!=\!1$ and $cost(i,dom)\!=\!1$ uniformly for all $i \!\in \!\mathcal{I}$ and $dom \!\in\! Dom$, the repair with the change tuple $\langle 4, (c,legal(p,r) \rangle$ of cost~1 is a solution to the MRP, and $p$ can now weakly win the game (by playing $r$).
\end{example}
\subsubsection{Existence of a Solution Repair.} Repair tasks (or MRPs) may have no solutions, for example, if the GTL formulas are contradictory, or when $|G_{E}|$ is too small to allow for a repair. The following theorem states sufficient conditions to ensure that a game can always be repaired to be n-well-formed for any $n>0$.% and \yifanm{with a desired maximal horizon $N$}.
% property can always be repaired. %to be well-formed. %(regardless of the cost). % of a $\gd$ can always be repaired.% w.r.t. . %with horizon $N$
\begin{theorem}
\label{theorem:exist}
Let $G$ be a $\gd$ and $R$ be the set of players in $G$. For any $n>0$, there exists a repair $\mathcal{R}$ on $G$ such that $rep(G,\mathcal{R})$ is n-well-formed if \textbf{all} of the following hold:
    \begin{itemize}
        \item $G \cup \Sinit^{true} \not\models \gdl{terminal}$, and for all $p_{i} \in R$, $|\gamma(p_{i})| \geq 2$.
        \item For each $p_{i} \in R$, there exists a state $S_{i} \subseteq \beta$ such that $G \cup S^{true}_{i} \models \gdl{goal}(p_{i},100) \wedge \gdl{terminal}$.
        \item $|G_{L}| + |G_{E}| \geq 2 \cdot |R|$ and $|G_{N}| + |G_{E}| \geq |\mathcal{N}| \cdot |R|$ as well as $|G_{L}| + |G_{N}| + |G_{E}| \geq (2+|\mathcal{N}|) \cdot |R|$.
    \end{itemize}
\end{theorem}
    % Let $\mathcal{T}=\langle G,\Phi^{+},\Phi^{-},cost\rangle$ with $\gd$~$G$, $R$ be the set of players in $G$, $\Phi^{-}=\{\psi_{loss}(p,N) \mid p \in R\}$ and $\Phi^{+}=\{\psi_{end}(N),\psi_{play}(N)\}$ for some $N>0$. 
    % Repair task $\mathcal{T}$ always has a solution if \textbf{all} of the following hold:
    % \begin{itemize}
    %     \item $G \cup \Sinit^{true} \not\models \gdl{terminal}$, and for all $p_{i} \in R$, $|\gamma(p_{i})| \geq 2$.
    %     \item For each $p_{i} \in R$, there exists a state $S_{i} \subseteq \beta$ such that $G \cup S^{true}_{i} \models \gdl{goal}(p_{i},100) \wedge \gdl{terminal}$.
    %     \item $|G_{L}| + |G_{E}| \geq 2 \cdot |R|$ and $|G_{N}| + |G_{E}| \geq |\mathcal{N}| \cdot |R|$ as well as $|G_{L}| + |G_{N}| + |G_{E}| \geq (2+|\mathcal{N}|) \cdot |R|$.
    % \end{itemize}
\begin{sketch}[Proof (Sketch).] We show that the conditions ensure that $G$ can always be repaired to be 1-well-formed, hence, it is n-well-formed for any $n > 0$. Note that \emph{any\/} $\gd$ in restricted form with at most $|G_{L}|+|G_{E}|$ legal, at most $|G_{N}|+|G_{E}|$ next, and at most $|G_{C}|$ legal or next rules and with the same set of base propositions, move domain, and other rules $G_{R}$, are obtainable from $G$ by a valid repair. If all conditions hold, there is a $\gd$ with $2 \cdot |R|$ legal rules and at most $|\mathcal{N}|\cdot |R|$ next rules so that all \textbf{1-max} sequences end in one of the $|R|$ terminal states, with the $i$-th one a winning state for $p_{i}$.
\end{sketch}
% The proof sketch implies that the 3 conditions above can also ensure that for any $N>0$, the \emph{repair task} defined in Example~\ref{example:repair} (i.e., repairing a game to be well-formed \emph{and} with a horizon at most $N$) must have a solution.

A repair task may have no solution if $|G_{E}|$ is chosen to be too small. The following theorem provides sufficient conditions to confirm that a repair task has no solution, even for arbitrarily large $|G_{E}|$\/: If $\psi_{end}(n) \!\in\! \Phi^{+}$ for some $n$ (i.e., the desired maximal horizon of the resulting game so that all $n$-max sequences terminate or end in a non-playable state) and the repair task has no solution with $|G_{E}|$ reaching some polynomial bound, increasing $|G_{E}|$ (i.e., allow for more new rules to be added) is not enough to make the task solvable.
\begin{theorem}
\label{theorem:unsat}
Let $\mathcal{T}\!=\!\langle G,\Phi^{+},\Phi^{-},cost\rangle$ be a repair task with $\psi_{end}(n) \!\in\! \Phi^{+}$ for some $n$. Let $K\!=\!max(1,|\Phi^{-}|)$ and $R$  the set of players in $G$.
If $\mathcal{T}$ has no solution and $|G_{E}|\!\geq\! K \cdot (n+1) \cdot |\mathcal{L}|+ n \cdot K\cdot |\mathcal{N}| \cdot (|R| + 1)$, then any repair task $\mathcal{T'}\!=\!\langle G',\Phi^{+},\Phi^{-},cost'\rangle$, where $G'\!=\!G_{L} \cup G_{N} \cup G_{E}' \cup G_{R}$ \yifanm{and $|G_{E}'| \geq |G_{E}|$}, also has no solution. 
\end{theorem}  
%For space reasons, a proof is provided in the appendix.
% The 
% if $|G_{E}|$ is too small. 
% The proof is available in the supplementary material. 
% The idea is that if $\mathcal{T}$ has a solution, there is a $\gd$ with at most $K \!\cdot\! (N+1)\! \cdot\! |\mathcal{L}|$ legal rules and at most $N\! \cdot\! K\! \cdot\! |\mathcal{N}|\! \cdot\! (|\mathcal{L}|+2)$ next rules describing a game with at most $K\! \cdot\! (N+1)$ states 
% that satisfies (resp. unsatisfies) all properties in $\Phi^{+}$ (resp. $\Phi^{-}$).    
\subsection{Complexity Results}
After having formally defined the game description repair problem, we now present some complexity results related to MRP by considering different decision problems. 
\begin{itemize}
    \item $\mrpb$: Given a repair task $\mathcal{T}$ and a cost \textbf{bound} $C \in \mathbb{N}$, decide if there is a solution repair of cost at most $C$. %\yifan{Eliminated MRPV, because it's irrelevant to MRP}
    %\item $\mrpv$: Given a repair task $\mathcal{T}$ and a valid repair $\mathcal{R}$ as input, \textbf{verify} if $\mathcal{R}$ is a lowest-cost solution repair to $\mathcal{T}$.
    \item $\mrpt$: For a repair task $\mathcal{T}$ and change \textbf{tuple} $t=\langle i,dom \rangle$% ($i \in \mathcal{I}$, $dom \in Dom$)
    , decide if some lowest-cost solution repair to $\mathcal{T}$ contains $t$.
\end{itemize}
\begin{theorem}
\label{brp:sigma2}
   $\mrpb$ is  $\Sigma_{2}^{P}$-complete.
\end{theorem}
\begin{proof} 
    (Membership) We can guess a repair $\mathcal{R}$, validate if $\mathcal{R}$ is valid with cost $\leq C$, and calculate $G'=rep(G,\mathcal{R})$ in PTIME. We can check if $G' \not\models_{t} \phi_{i}^{-}$ for all $\phi_{i}^{-} \in \Phi^{-}$ in NP time,   
    % and a $deg(\phi_{i})$-max sequence $(\Sinit^{i},\myDots,S_{n_{i}}^{i})$ with $S_{0}=S_{0}^{i}$ for all $1 \leq i \leq |\Phi^{-}|$. Let $G'=rep(G,\mathcal{R})$. We can validate if $\mathcal{R}$ is valid with a cost $\leq C$, and if $G',(\Sinit^{i},\myDots,S^{i}_{n_{i}}) \not\models_{t} \phi^{-}_{i}$ for all $\phi_{i}^{-} \in \Phi^{-}$ in PTIME 
    and if $G' \models_{t} \phi^{+}_{i}$ for all $\phi^{+}_{i} \in \Phi^{+}$ in co-NP time ($\theoremshort$~\ref{theorem:basecase}). Thus, $\mrpb$ is in $\text{NP}^{\text{NP}} = \Sigma_{2}^{P}$. 
    
(Hardness) As a $\Sigma_{2}^{P}$-hard problem, we reduce deciding the validity of a quantified Boolean formula (QBF) of the following form~\cite{stockmeyer1976polynomial} to an $\mrpb$\/:
    \begin{equation}
    \label{qbf}
        \Psi=\exists x_{1}\myDots x_{n} \forall y_{1} \myDots y_{m} \ \ E,  \ \ n, m \geq 1  \tag{E1}
    \end{equation}
    where $E=D_{1} \vee \myDots \vee D_{k}$, each $D_{i}=l_{i}^{1} \wedge l_{i}^{2} \wedge l_{i}^{3}$, and each $l_{i}^{j}$ is a literal over variables $\existblock \cup \univblock$ with $\existblock=\{x_{1},\myDots,x_{n}\}$ and $\univblock=\{y_{1},\myDots,y_{m}\}$. %W.l.o.g.\ we assume that each $D_{i}$ contains at least 1 universal variable. 

    Let $G$ be a $\gd$ with players $p_{1},\myDots,p_{m}$, base propositions $r_{1},x_{1},\myDots,x_{n},y_{1}^{+},\myDots,y_{m}^{+},y_{1}^{-},\myDots,y_{m}^{-}$, and the move domain of each player being $\{pos,neg\}$. $G$ contains the following rules with $G_{L}=R_{1} \cup R_{2}$, $G_{N}=R_{3} \cup \ldots \cup R_{6}$, $|G_{E}|=0$, and $G_{R}$ is $R_{7}$ plus the definitions of $\gdl{role}$, $\gdl{base}$, and $\gdl{input}$:
    \begin{itemize}
        \item $R_{1} = \bigcup_{i=1}^{m} \{[r_{i}]\ \ \gdl{legal}(p_{i},pos).\}$
        \item $R_{2}= \bigcup_{i=1}^{m} \{[r_{i+m}]\ \ \gdl{legal}(p_{i},neg).\}$
        \item $R_{3} = \bigcup_{i=1}^{n} \{[r_{i+2m}]\ \ \gdl{next}(x_{i}) \texttt{:-} \gdl{true}(x_{i}).\}$
        \item $R_{4} = \bigcup_{i=1}^{m} \{[r_{i+2m+n}]~\gdl{next}(y^{+}_{i}) \texttt{:-} \gdl{does}(p_{i},pos).\}$
        \item $R_{5} = \bigcup_{i=1}^{m} \{[r_{i+3m+n}]~\gdl{next}(y^{-}_{i}) \texttt{:-} \gdl{does}(p_{i},neg).\}$
        \item \yifanm{$R_{6} = \{[r_{4m+n+1}]~next(r_{1}).\}$}
        \item $R_{7}=\bigcup_{i=1}^{k} \{\gdl{terminal} \texttt{:-} \sigma(l_{i}^{1}),\sigma(l_{i}^{2}),\sigma(l_{i}^{3}),true(r_1).\}$
    \end{itemize}
    %    $\bullet$ \\
    %    $\bullet$ \\
    %    $\bullet$ \\
    %    $\bullet$ \\
     \ \ \ \ where $\sigma(l_{i}^{j})\!\!=\!\!\begin{cases} 
        \gdl{true}(y_{s}^{+})& \!\!\!\text{if}~l_{i}^{j}=y_{s}~\text{and}~y_{s} \in \univblock\\
        \gdl{true}(y_{s}^{-})& \!\!\!\text{if}~l_{i}^{j}=\neg y_{s}~\text{and}~y_{s} \in \univblock\\
        \gdl{true}(x_{s}) & \!\!\!\text{if}~l_{i}^{j}=x_{s}~\text{and}~x_{s} \in \existblock \\
        not~\gdl{true}(x_{s})& \!\!\!\text{if}~l_{i}^{j}=\neg x_{s}~\text{and}~x_{s} \in \existblock \\
    \end{cases}$
    
        % Consider a $\mrpb$ of form $\langle G,\Phi^{+},\Phi^{-},cost,C\rangle$. Here
    Let $C=2^{n}\!-\!1$, $\Phi^{-}=\{\}$, and $\Phi^{+}=\{\X terminal\}$, and 
    $\cost(i,(-,l))=2^{2m+n-i}$ when $2m\!+\!1 \leq i \leq2m\!+\!n$ and $l \in bd(r_{i})$; otherwise, $cost(i,(tp,l))=2^{n}$.
    %Namely, $G$ has rules of form $R_{1}$ to $R_{6}$.
    Put in words,
    removing $true(x_{i})$ from the body of the $i$-th rule of $R_{3}$  ($1 \leq i \leq n$)  costs $2^{n-i}$ (i.e., $2^{n-1},\myDots,2^{1},2^{0}$) while any other change costs $2^{n}$. We require the repaired $\gd$ to satisfy $G' \models_{t} \X terminal$ and the repair must cost $\leq 2^{n}-1$. Note that the $\mrpb$ $\langle G,\Phi^{+},\Phi^{-},cost,C \rangle$ can be constructed in PTIME w.r.t. the size of $\Psi$. $\Psi$ is valid iff the $\mrpb$ is true\/: %Since the repair costs $\leq \!2^{n}\!-\!1$, checking $\X terminal$ is equivalent to checking if all 1-max sequences $(\Sinit,S_{1})$ in $G'$ terminate at $S_{1}$.
    
    ``$\Rightarrow$'': If $\Psi$ is valid, there exists $X_{T} \subseteq \existblock$ such that if we set all $x \in X_{T}$ to true and all $x \in \existblock \setminus X_{T}$ to false, then for any $Y_{T} \subseteq \univblock$ if we assign all $y \in Y_{T}$ to true and all $y \in \univblock \setminus Y_{T}$ to false, some $D_{j} \in E$ must be satisfied. Consider a repair where $\langle i\!+\!2m,(-,true(x_{i}) \rangle$ is a change tuple in $\mathcal{R}$ (i.e., we remove $\gdl{true}$$(x_{i})$ from the $i$-th rule of $R_{3}$) iff $x_{i} \in X_{T}$. The repair costs $\leq \!2^{n}\!-\!1$, and $\gdl{terminal}$ holds in the next state \yifanm{for any joint action $A$. Concretely, if %$A^{does}\!=\!\{does(p_{1},a_{1}),\!\myDots\!,does(p_{m},a_{m})\}$ 
    $A^{does}\!=\!\bigcup_{i=1}^{m}\{does(p_{i},a_{i}).\}$ 
    where $a_{i}=pos$ iff $y_{i} \in Y_{T}$, then the $j$-th rule of form $R_{7}$ in the repaired $\gd$ can be activated in the next state $S_{1}$ where $S_{0}~\xrightarrow[]{A}~S_{1}$}. 
    
    ``$\Leftarrow$'': Similarly, if the $\mrpb$ has a solution $\mathcal{R}$, then assigning $x_{i}\!=\!\top$ iff $\langle i\!+\!2m,(-,true(x_{i}) \rangle \in \mathcal{R}$, gives a satisfiable partial assignment to the existential variables in~$\Psi$.   
    %Deciding $\Psi$ is $\Sigma_{2}^{P}$-hard, so is $\mrpb$.
\end{proof} %(Membership) Since the maximum possible cost of a valid repair is bounded ($\defshort$~\ref{def:repset} and~\ref{def:cost}) and the answer to $\mrpb$ is monotonic, we can binary search for the lowest cost bound $C$ such that the $\mrpb$ returns ``yes'' and see if $C$ is odd. The number of calls to the $\mrpb$ oracle is polynomial w.r.t. the size of the input. By $\theoremshort$~\ref{brp:sigma2}, MRP is in $\Delta_{3}^{P}$. 
\begin{theorem}
\label{mrp:tuple}
    $\mrpt$ is $\Delta_{3}^{P}$-complete.
\end{theorem}
\begin{sketch}[Proof (Sketch).] (Membership) Since the maximum possible cost of a valid repair is bounded ($\defshort$~\ref{def:repset} and~\ref{def:cost}) and the answer to $\mrpb$ is monotonic, we can do a binary search for the lowest cost bound $C$ such that the $\mrpb$ returns ``yes''. Once we have that $C$, consider a new cost function: $cost'(j,dom')=2 \cdot cost(j,dom')$ if $\langle j,dom' \rangle \neq \langle i,dom \rangle$, and $cost'(i,dom) \!= \!2 \cdot cost(i,dom) \!- \! 1$. The answer to the $\mrpt$ is the same as $\mrpb$ with the new cost function and a cost bound $2 \! \cdot \! C \! - \!1$. 
Since the number of $\mrpb$ oracle calls is polynomial w.r.t. the size of the input, $\mrpt$ is in $\Delta_{3}^{P}$. 

(Hardness) For a TQBF~\cite{krentel1992generalizations} of the form~(\eqref{qbf}), we create an $\mrpt$ with the same $G$, $\Phi^{+}$, $\Phi^{-}$, and cost function as in $\theoremshort$~\ref{brp:sigma2}. The $\Delta_{3}^{P}$-hard problem ``Deciding if the lexicographically smallest satisfiable partial assignment to the existential variables of a TQBF of the form (\ref{qbf}) has $x_{n}=\top$'' can be reduced to checking if the change tuple $\langle n\!+\!2m,(-,true(x_{n})) \rangle$ is in some optimal repair.
\end{sketch}
The complexity class of the actual function problem MRP is given as follows, and the proof can be easily obtained by modifying the proof of Theorem~\ref{mrp:tuple}.
\begin{theorem}
\label{mrp:optimization}
    MRP is $F\Delta_3^{P}$-complete. 
\end{theorem}

We conclude by pointing out that solving an MRP with $\Phi^{+}=\{\}$ is simpler; e.g., fixing weak winnability only can be achieved by $\Phi^{+}=\{\}$ and the same $\Phi^{-}$ as in Example~\ref{example:repair} for some $n>0$, so that for every player, there is an $n$-max sequence that terminates with $goal(p,100)$. 
 \begin{theorem}
 \label{mrp:cheap}
     $\mrpb$ is NP-complete, %$\mrpv$ $\text{co-NP}$-complete, 
     $\mrpt$ $\Delta_{2}^{P}$-complete, and MRP $F\Delta_{2}^{P}$-complete if $\Phi^{+}=\{\}$.
 \end{theorem}
%The proof of Theorem~\ref{brp:np} uses techniques similar to those in Theorems~\ref{brp:sigma2} and~\ref{MRP}; the details are provided in the supplementary material. 
%The above theorems indicate that repairing termination and playability is computationally more expensive than repairing weak winnability because weak winnability can be expressed with
%\munyque{Why? Is it because weak winnability can be expressed with $|\Phi^{+}|=0$?}
% \begin{sketch}[Proof (Sketch).]
%         (Membership) One can guess a repair set $\mathcal{R}$ and a $deg(\phi_{i})$-max sequence for each $1 \leq i \leq |\Phi^{-}|$. One can check if $\mathcal{R}$ costs $\leq C$ and $rep(G,\mathcal{R}),(\Sinit,\myDots,S^{i}_{n_{i}}) \not\models_{t} \phi_{i}$ for each $\phi_{i} \in \Phi^{-}$ in PTIME. (Hardness) The NP-hard problem ``deciding if a 1-player GDL game with a horizon $N$ encoded in unary has a solution plan''~\cite{bonnet2014complexity} can be reduced to a BRP with $|\Phi^{+}|=0$.
% \end{sketch}

%% file: encoding-proof.tex
\section{Encoding}
\label{sec:encoding}
We discuss an ASP-based approach to solving minimal repair problems. Theorem~\ref{mrp:optimization} shows that MRP is $F\Delta_{3}^{P}$-complete, beyond the expressiveness of normal logic programs with optimization statements ($F\Delta_{2}^{P}$). Disjunctive logic programs with optimization statements, however, can express $F\Delta_{3}^{P}$ problems~\cite{javier2025}.
Hence, we can encode an
MRP by a disjunctive ASP with optimization statements using the saturation technique~\cite{eiter1995computational}, and extract the change tuples as well as the repaired $\gd$ from the stable model of the program.

With the help of an existing \textit{guess and check} tool\footnote{Available at %Potassco link:
\url{https://github.com/potassco/guess_and_check/}}, a suitable disjunctive ASP can be automatically generated rather than having to be written from scratch. 
%\subsubsection{Guess and Check.}
A \emph{guess and check} (G\&C) program~\cite{eiter2006towards} is a pair of logic programs $\langle P_{G},P_{C} \rangle$. $\mathcal{M}$ is a \emph{stable model of $\langle P_{G},P_{C} \rangle$\/} iff $\mathcal{M}$ is a stable model of $P_{G}$, and $P_{C}\cup\mathcal{H}$ is unsatisfiable, where $\mathcal{H}$ is the set of atoms of the form $holds(x)$ in $\mathcal{M}$. The Potassco software translates a G\&C program into a disjunctive program and solve it. 
To solve MRP using G\&C, we map rules in the $G_{C}$ part of a $\gd$ to ASP atoms as follows. 

\begin{definition}\label{def:mapping} Let $G$ be a valid $\gd$. \yifanm{We define a mapping $\tau$ from literals in $G_{C}$ to tuples: for each} $q \in \emptyset \cup \mathcal{A} \cup \mathcal{F} \cup \mathcal{L} \cup \mathcal{N}$,
\begin{itemize}
    \item $\tau(q)=\emptyset$ if $q=\emptyset$
    \item $\tau(q)=(ba,f)$ if $q=next(f)$, for some $f$
    \item $\tau(q)=(ac,(p,a))$ if $q=legal(p,a)$, for some $p,a$
    \item $\tau(q)=(pos,ba,f)$ if $q=true(f)$, for some $f$
    \item $\tau(q)=(neg,ba,f)$ if $q=not~true(f)$, for some $f$
    \item $\tau(q)=(pos,ac,(p,a))$ if $q=does(p,a)$, for some $p,a$
    \item $\tau(q)=(neg,ac,(p,a))$ if $q=not~does(p,a)$, for $p,a$
\end{itemize}
We define the inverse mapping from tuples to GDL atoms as $\tau^{-1}$ such that $\tau^{-1}(\tau(q))=q$ and $\Pi$ be the mapping from rules in $G_{C}$ to a \emph{\bf set} of ASP atoms of the form $ha()$ and $lit()$: 
% \begin{itemize}
%     \item $\Pi(G_{C})=\bigcup_{i \in \mathcal{I}}$
%     \item $\{he(i,\tau(q)) \mid hd(r_{i})=q\} \cup \{lit(i,\tau(q)) \mid q \in bd(r_{i})\}$
% \end{itemize}
    \begin{itemize}
        \item $\hd(i,\tau(q)) \in \Pi(G_{C})$ iff $hd({r_{i}})=q$ for $i \in \mathcal{I}$
        \item $lit(i,\tau(q)) \in \Pi(G_{C})$ iff $q \in bd(r_{i})$ for $i \in \mathcal{I}$
    \end{itemize}
    %$\bullet$ $hd(i,\tau(q)) \in \Pi(G_{C})$ iff $hd({r_{i}})=q$
    %$\bullet$ $lit(i,\tau(q)) \in \Pi(G_{C})$ iff $q \in bd(r_{i})$\\
$\Pi^{-1}$ is the inverse mapping such that $\Pi^{-1}(\Pi(G_{C}))=G_{C}$.
\end{definition}
\noindent
For example, $c_{2}$ in Fig.~\ref{example:repairset} is mapped to $\hd(2,(ba,loss))$ and $lit(2,(pos,ac,(p,l)))$, and $\Pi^{-1}$ maps them back to $c_{2}$. 

First, we create a program $P_{\mathcal{D}}$ defining the \textbf{domain} of: the IDs of old rules $1,\ldots,|G_{O}|$ ($o\_rule$), the IDs of rules in $G_{C}$ ($rule$), the IDs of empty rules ($e\_rule$), the polarity of literals pos/neg ($pol$), and the set of atomic propositions ($atom$), which contains $atom(ac,(p,a))$ for each $(p,a)$ in the move domain and $atom(ba,f)$ for each $f$ as a base proposition. To avoid naming conflicts, we use $o\_\hd$ (old head) and $o\_lit$ when referring to the ASP representation (cf.~$\defshort$~\ref{def:mapping}) of the input $G$ (i.e., $\Pi(G_{C})$), and $lit$ and $hd$ for the repaired $G'$. 

Based on $P_{\mathcal D}$ we create a \textbf{generator} $\pgen(G)$ for the MRP,
\begin{itemize}
    \item $\pgen(G)= P_{\mathcal{D}} \cup \Pi(G_{C}) \cup \{\eqref{enc:rule1},\ldots,\eqref{enc:rule14}\}$ (cf.~Fig.~\ref{fig:program})
\end{itemize}
to manipulate the ASP representation of all the rules $G_{C}'$ that can be obtained from $G$ with some valid repair. We also create a $\gd$ $\ginv$ (not in restricted form) to \textbf{simulate} $\Pi^{-1}$:
\begin{itemize}
    \item $\ginv=\{\eqref{enc:rule16}, \ldots, \eqref{enc:rule21}\}$ (cf.~Fig.~\ref{fig:program})
%    \michael{any reason we need to use $\cup..\cup$?)} \yifan{our GDL input is a multiset} %\yifan{no particular reason. Just the input GDL is a multiset, I want people to think $\cup$ is like a concatenation instead of union}
\end{itemize}
which ensures that if $\mathcal{X}=\Pi(G_{C}')$ for some repaired $\gd$ $G'$ obtained from $G$, $\ginv \cup \mathcal{X} \cup G_{R}$ is equivalent $\Pi^{-1}(\mathcal{X}) \cup G_{R}$.

Before we discuss the G\&C encoding, let's take a closer look at $\pgen$ and $\ginv$. Clauses~\eqref{enc:rule1}-\eqref{enc:rule1.1} specify that every rule in $G_{E}$ can remain empty or be modified to a legal/next rule. We use $rtype(i,ba)$ (resp.\ $rtype(i,ac)$) to denote the \emph{type of} the $i$-th rule as a next (resp.\ legal) rule. Clauses~\eqref{enc:rule5}-\eqref{enc:rule6} state that every rule in $G_{O}$ can be deleted (i.e., its head changed to~$\emptyset$) or acquire a new head of the \emph{same type}. Clause~\eqref{enc:rule7} says that if the new head $hd(r_{i}')$ of a rule differs from the original one, the repair must contain the change tuple $\langle i, (c,hd(r'_{i})) \rangle$. We use $tup(i,tp,l)$ to denote the change tuple $\langle i,(tp,\tau^{-1}(l)) \rangle$ (cf.\ \defshort~\ref{def:mapping}).

Clauses~\eqref{enc:rule8}-\eqref{enc:rule12} model repairs to the rule bodies after we have fixed the heads.
For every rule $r_{i}$ that does not have a dummy head in the resulting $\gd$ (i.e., $hd(r_{i}') \neq \emptyset$), \eqref{enc:rule8}~says that we can remove any literal $l$ in the body with the change tuple $\langle i,(-,l) \rangle$; \eqref{enc:rule9}~says that we can add any $true(f)$ or its negation to the body if $f \in \beta$; and \eqref{enc:rule10}~says that we can add any $does(p,a)$ or its negation to the body if $(p,a) \in \gamma$ and $r_{i}$ is a next rule. Clauses~\eqref{enc:rule11}-\eqref{enc:rule12} ensure that if $hd(r_{i}')=\emptyset$, $bd(r_{i}')=\{\}$ (cf.\ \defshort~\ref{def:repset}). Otherwise, $bd(r_{i}')$ in the resulting $\gd$ contains all literals that have been added to $r_{i}$ and all literals in $bd(r_{i})$ in the input $\gd$ that have not been removed.

We also introduce constraints to improve the quality of the repair (clauses \eqref{enc:rule13}--\eqref{enc:rule14})\/: Rules with an atom that appears positively and negatively ($a$ and $not~a$) in the body, or with two different actions for the same player (i.e., $does(p,a)$ and $does(p,b)$ with $a \neq b$), are redundant, hence not allowed.
% define constraints to avoid these trivial redundancies in the resulting $\gd$, thus improving the quality of the repair. 
% violating these constraints is redundant.
% %For example, %

% In the resulting $\gd$, \eqref{enc:rule13} says that no atom can appear  of a rule in the resulting $\gd$, and \yifanm{\eqref{enc:rule14} says that positive atoms  cannot appear in the body of the same rule. }% in the resulting $\gd$}. 
 %\munyque{It is better to also explain \eqref{enc:rule14}.} 
%\yifan{\eqref{enc:rule14} says $does(p,a)$ and $does(p,b)$ where $a \neq b$ cannot appear simultaneously in the body of a rule, it's related to GDL semantics, I found it hard to explain concisely given the space limit.} 
%\munyque{I changed the end of the paragraph as follows: }
%Note that, in practice, no $\gd$ has these redundant rules. Even if our main objective is to find an optimal repair, we also want the repaired $\gd$ to maintain high quality. Hence, we introduce constraints to reduce trivial redundancies in the resulting $\gd$. 
\input{encoding}
To sum up, due to the way $\pgen(G)$ is encoded, by Definition~\ref{def:repset}, any stable model of $\pgen(G)$ corresponds to a valid $\gd$ $G'$ that can be obtained from $G$ with some valid repair. The predicates $lit$ and $hd$ record the ASP representation of $G_{C}'$, and the predicate $tup$ records all the change tuples. 

We now prove that $\ginv$ models $\Pi^{-1}$. If $\mathcal{X}=\Pi(G_{C}')$ for some repaired $\gd$ $G'$, then $G''=\ginv \cup \mathcal{X} \cup G_{R}$ is equivalent to $G'=\Pi^{-1}(\mathcal{X}) \cup G_{R}$ in the sense that at any state, a ground atom holds in $G'$ iff it holds in $G''$. %\michael{why use ":"?} \yifan{fixed to ``=''} 
\begin{theorem}
\label{theorem:inv}
    Let $G$ be a valid $\gd$, $\mathcal{M}$ a stable model of $\pgen(G)$, and $\mathcal{X}$ the set of all atoms $lit,ha$ in $\mathcal{M}$. If $q(\vec t\,)$ is a ground atom with a predicate symbol in $\{true,does,$ $legal,next\}$ or appears in $G_{R}$, then for any state $S$ and joint action $A$ over $G$'s base propositions and move domain\/:
    \begin{itemize}
    %\setlength\itemsep{0.2em}
     %   \item If $q \notin \{next,does\}$, then $\Pi^{-1}(\mathcal{X})  \cup G_{R} \cup S^{true} \models q(\vec t\,)$ iff $\ginv \cup \mathcal{X} \cup G_{R} \cup S^{true} \models q(\vec t\,)$.
      %  \item If $q \in \{next,does\}$, then $\Pi^{-1}(\mathcal{X}) \cup G_{R} \cup S^{true} \cup A^{does} \models q(\vec t\,)$ iff $\ginv \cup \mathcal{X} \cup G_{R} \cup S^{true} \cup A^{does} \models q(\vec t\,)$.
          \item $\Pi^{-1}(\mathcal{X}) \cup G_{R} \cup S^{true} \cup A^{does} \models q(\vec t\,)$ iff $\ginv \cup \mathcal{X} \cup G_{R} \cup S^{true} \cup A^{does} \models q(\vec t\,)$.
    \end{itemize}
  % \michael{Is it really necessary to distinguish the two cases? Does it ``hurt'' if $A^{does}$ were added in the first case?} \yifan{Fixed}
\end{theorem}
\begin{sketch}[Proof (Sketch).]
    Let $G'=\Pi^{-1}(\mathcal{X}) \cup G_{R}$ and $G''=\ginv \cup \mathcal{X} \cup G_{R}$.
    By construction of $\pgen(G)$, $G'$ must be a valid $\gd$ obtained from $G$ with some valid repair.
    There are four cases. First, if $q \in \{true,does\}$, the statement trivially holds.
    Second, if $q \in G_{R}$, the statement holds because $q(\vec t\,)$ solely depends on $S^{true}$ and $G_{R}$ since, as $G$ is in restricted form, $q$ %does not depend on $does$ and
    cannot appear as head in $\ginv \cup \mathcal{X}$ or $\Pi^{-1}(\mathcal{X})$. 
    Third, if $q(\vec t\,)=legal(p,a)$ for some $p,a$. Since both $G'$ and $G''$ are valid $\gd$s, whether $legal(p,a)$ holds or not 
    does not depend on $A^{does}$. Thus, $G' \cup S^{true} \models q(\vec t\,)$ iff $\hd(i,(ac,(p,a))) \in \mathcal{X}$ for some $i$ and there is no $f$ such that: i) $true(f) \in S^{true}$ and $lit(i,(neg,ba,f)) \in \mathcal{X}$, \textbf{or} ii) $true(f) \notin S^{true}$ and $lit(i,(pos,ba,f)) \in \mathcal{X}$. This is exactly what clauses \eqref{enc:rule16}--\eqref{enc:rule17} and \eqref{enc:rule20} in Fig.~\ref{fig:program} are modeling\/: \eqref{enc:rule16} and \eqref{enc:rule17} state that $err\_t(i)$ is justified iff for some $f$ and $i$: i) $true(f) \in S^{true}$ and $lit(i,(neg,ba,f)) \in \mathcal{X}$, \textbf{or} ii) $true(f) \notin S^{true}$ and $ lit(i,(pos,ba,f)) \in \mathcal{X}$. \eqref{enc:rule20} enforces $G''  \cup S^{true} \models q(\vec t\,)$ iff there is some $i$ such that $\hd(i,(ac,(p,a))) \in \mathcal{X}$ and $err\_t(i)$ is not justified. The final case when $q=next$ is analogous to the previous case (see \eqref{enc:rule16}--\eqref{enc:rule19} and \eqref{enc:rule21}). %say that $next(f)$ is activated iff no atoms of $\gdl{does}$ or $\gdl{true}$ in the body of some rule with head $next(f)$ are unjustified. 
\end{sketch}
Based on Theorem~\ref{theorem:basecase} and~\ref{theorem:inv}, we can check if the repaired game description $G'=G_{C}'\cup G_{R}$, with  $\mathcal{X}=\Pi(G_{C}')$ generated by $\pgen(G)$, satisfies a given GTL property as follows. 
\begin{corollary}
\label{theorem:correct}
     Let $G$ be a valid $\gd$, $\phi$ be a GTL formula with $deg(\phi)=n$, $\mathcal{M}$ a stable model of $\pgen(G)$, and $\mathcal{X}$ the set of all atoms of $lit,\hd$ in $\mathcal{M}$. Let $G'=\Pi^{-1}(\mathcal{X}) \cup G_{R}$, and  
    %, and $\mathcal{M}$ be a stable model of $D \cup \pgen \cup \Pi(G_{C})$. 
    $\pver(\phi)=\plegal^{n}   \cup \pext^{n}(\ginv \cup G_{R}) \cup \enc(\phi,0) \cup \{\texttt{:-} \eta(\phi,0)\}$. Then, $G' \models_{t} \phi$ iff $\mathcal{X} \cup \pver(\phi)$ has no stable model. % and $G' \not\models_{t} \phi$, otherwise. 
\end{corollary}
We return to the G\&C framework. Consider an MRP instance $\langle G,\{\phi^{+}_{1},\ldots,\phi^{+}_{m}\},\{\phi^{-}_{1},\ldots,\phi^{-}_{n}\},cost\rangle$. It is trivial that in GTL, $G \models_{t} \phi_{i}^{+}$ holds for all $1 \leq i \leq m$ iff $G \models_{t} \phi_{0}$, where $\phi_{0}=\phi^{+}_{1} \wedge \ldots \wedge \phi^{+}_{m}$. Thus, we can replace formulas in $\Phi^{+}$ with the single formula $\phi_{0}$. For the G\&C, we define:
\begin{itemize}
    \item $P_{G}=\{\eqref{enc:optimization}\} \cup \pgen(G)\cup P_{H}  \cup \bigcup_{i=1}^{n} \pver(\phi_{i}^{-},i)$
    %\item $\mathcal{H}$: all instances of $lit$ and $\hd$ in the stable model of $P_{G}$
    \item $P_{C}=\pver(\phi_{0}) \cup P'_{H}$ (cf.\ Corollary~\ref{theorem:correct}).
\end{itemize}
where $P_{H}=\{holds(q(I,F))\texttt{:-}q(I,F). \mid q \in \{lit,\hd\}\}$ %\cup \{holds(lit(I,F))\texttt{:-}lit(I,F)\}$
and
    $P'_{H}=\{q(I,F)\texttt{:-}holds(q(I,F)). \mid q \in \{lit,\hd\}\}$.
    
    %\item $P_{L_{1}}=\{holds(lit(I,F))\texttt{:-}lit(I,F).\}$
    %\item $P_{H_{2}}=\{hd(I,F))\texttt{:-}holds(hd(I,F)).\}$ %\cup \{holds(lit(I,F))\texttt{:-}lit(I,F)\}$
    %\item $P_{L_{2}}=\{lit(I,F))\texttt{:-}holds(lit(I,F)).\}$
    %\item $P_{H_{2}}=\{P(I,F)\texttt{:-}holds(P(I,F)). \mid P\in \{\hd,lit\}\}$ 
    % \item $P(\phi,j)=\{\texttt{:-} \eta(\phi,0)(j)\} \cup Ext^{deg(\phi)}_{j}(G_{R}) \cup  P^{legal}_{j,deg(\phi)} \cup Enc_{j}(\phi,0)    \cup Ext^{deg(\phi)}_{j}(\ginv)$
    %\item $P(\phi,j)\!=\!copy(P^{legal}_{n}   \cup Ext^{n}(\mathcal{X} \cup \ginv \cup G_{R}) \cup Enc(\phi,0) \cup \{\texttt{:-} \eta(\phi,0)\}, j)$
    %\item $P(\phi,j)\!=\!copy(P^{legal}_{n}   \cup Ext^{n}(\mathcal{X} \cup \ginv \cup G_{R}) \cup Enc(\phi,0) \cup \{\texttt{:-} \eta(\phi,0)\}, j)$
%\end{itemize}
Here, $P_{G}$ consists of: First, the repair generator $\pgen(G)$ to generate the ASP representation of $\mathcal{X}=\Pi(G_{C}')$ of all possible resulting $\gd$s $G'=G'_{C} \cup G_{R}$. Second, an ASP optimization statement written as a weak constraint~\eqref{enc:optimization} (cf.\ Fig.~\ref{fig:program}) which says that the stable model of $P_{G}$ should minimize the total cost of the change tuples, ensuring that the repair output by $\pgen(G)$ is an optimal one. 

Third, $P_{G}$ uses a checker $\pver(\phi_{j}^{-},j)$ for every $\phi_{j}^{-}$ to ensure that $G' \not\models_{t} \phi^{-}_{j}$. $\pver(\phi^{-}_{j},j)$ modifies $\pver(\phi^{-}_{j})$ in Corollary~\ref{theorem:correct} by \emph{extending} each occurrence of $q(\vec t\,)$ with a constant $j$ to $q(\vec t ,j)$ if $q$ is a predicate that is neither $\hd$ nor $lit$. For example, $does(R,A,i)$ in $\plegal^{n}$ become $does(R,A,i,j)$ ($\defshort$~\ref{def:aspgenerator}). The extension is crucial because with $\pver(\phi)$, we can only show $G' \not\models_{t} \phi$ for a single $\phi$ by showing that $\mathcal{X} \cup \pver(\phi)$ has a stable model, which involves generating a $deg(\phi)$-max sequence such that $G',(\Sinit,\myDots,S_{k}) \not\models_{t} \phi$. In our case, we need to show $G' \not\models_{t} \phi^{-}_{j}$ for all $j \leq n$, which requires generating $n$ sequences simultaneously that do not interfere with each other, where the $j$-th one dissatisfies $\phi^{-}_{j}$. To do so, we create $n$ ``copies'' of $\pver(\phi)$ letting them \emph{share} the set $\mathcal{X}$ and the $j$-th copy generates a sequence that dissatisfies $\phi^{-}_{j}$
to prove $G' \not\models_{t} \phi^{-}_{j}$. 

Finally, $P_{H}$ in $P_{G}$ ``wraps'' all atoms of $lit$ and $\hd$ (i.e., atoms in $\mathcal{X}$) of a stable model of $P_{G}$ with the G\&C preserved predicate $holds$, allowing $G'$ to be \emph{shared} between $P_{G}$ and $P_{C}$. Let $\mathcal{H}$ denote the set of all atoms of the form $holds(x)$ in the stable model of $P_{G}$.
The $P_{H}'$ part in $P_{C}$ ``decodes'' every instances in $\mathcal{H}$ back to $lit$ and $\hd$ (i.e., the set $\mathcal{X}$). (Recall that in G\&C, the stable model of $P_{G}$ is the stable model of the overall program iff $\mathcal{H} \cup P_{C}$ is unsatisfiable.) This is equivalent to $\pver(\phi_{0}) \cup \mathcal{X}$ has no stable model, which proves that $G' \models_{t} \phi_{0}$ (cf. Corollary~\ref{theorem:correct}). 

To sum up, $P_{G}$ generates an answer to the MRP by only considering the constraints in $\Phi^{-}$. The repaired $\gd$ is shared between $P_{G}$ and $P_{C}$ via the special G\&C predicate $holds$. Letting $P_{C}$ together with the set of $holds$ instances to have no stable models ensures that all formulas in $\Phi^{+}$ hold, which in turn implies that the answer generated by $P_{G}$ is an answer to the overall MRP. Our construction also indicates that for MRP with $\Phi^{+}=\{\}$ (cf.~Theorem~\ref{mrp:cheap}), $P_{G}$ alone, a normal logic program with optimization statements, suffices to solve the problem.
% Our construction implicitly confirms that MRP with $\Phi^{+}\!=\!\{\}$ is simpler (Theorem~\ref{mrp:cheap}). $P_{G}$, a normal logic program with optimization statements, is sufficient to solve the problem. 
% One remark is that $\text{MRP}_{t}$ with 
% $\Delta_{2}^{P}$-complete

%% file: encoding.tex
\begin{figure*}[ht]
\footnotesize
\centering
\begin{enumerate}[label=(\arabic{enumi}):~,ref=(\arabic{enumi}),align=right,leftmargin=\widthof{(100)}+\labelsep]
\item \label{enc:rule1}  \texttt{1\{ha(I,(TP,F)):atom(TP,F);ha(I,$\emptyset$)\}1~:-~e\_rule(I)}. %\hspace*{3.2em}
\item  \label{enc:rule1.1}\texttt{rtype(I,TP) :- ha(I,(TP,F)),~e\_rule(I)}.

%\item   \label{enc:rule3} \texttt{tup(I,c,(TP,F))~:-~ha(I,(TP,F)),~e\_rule(I)}.   

%\item \label{enc:rule4} \texttt{\{tup(I,c,$\emptyset$)\}~:-~o\_rule(I)}. %\hspace*{17.0em}~~
%\item \label{enc:rule4.1} \texttt{ha(I,$\emptyset$) :-~o\_rule(I), tup(I,c,$\emptyset$)}.       

\item   \label{enc:rule5} \texttt{rtype(I,TP) :- o\_ha(I,(TP,F)).}

\item  \label{enc:rule6} \texttt{1\{ha(I,(TP,F)):atom(TP,F);ha(I,$\emptyset$)\}1 :-~rtype(I,TP), o\_rule(I)}.

\item \label{enc:rule7} \texttt{tup(I,c,F)~:-~rule(I),~ha(I,F),~not~o\_ha(I,F)}.

\item \label{enc:rule8} \texttt{\{tup(I,-,(Q,TP,F))\}~:- o\_lit(I,(Q,TP,F)),~not~ha(I,$\emptyset$)}.

\item \label{enc:rule9} \texttt{\{tup(I,+,(Q,ba,F))\}~:-~atom(ba,F),}~\texttt{pol(Q),}~\texttt{not~o\_lit(I,(Q,ba,F)),}~\texttt{rule(I),}~\texttt{not~ha(I,$\emptyset$)}. 

\item \label{enc:rule10}
                \texttt{\{tup(I,+,(Q,ac,F))\} :-~not~o\_lit(I,(Q,ac,F)),} \texttt{not~ha(I,$\emptyset$),}~\texttt{rtype(I,ba),}~\texttt{atom(ac,F),}~\texttt{pol(Q)}.

\item \label{enc:rule11} \texttt{lit(I,(Q,TP,F))~:-~tup(I,+,(Q,TP,F))}.             

\item \label{enc:rule12}
                    \texttt{lit(I,(Q,TP,F))~:-~o\_lit(I,(Q,TP,F)),~not~tup(I,-,(Q,TP,F)),~not~ha(I,$\emptyset$)}.

\item  \label{enc:rule13} \texttt{:-~lit(I,(pos,TP,F)),~lit(I,(neg,TP,F))}.        

\item   \label{enc:rule14} \texttt{:-~lit(I,(pos,ac,(P,A1))),~lit(I,(pos,ac,(P,A2))),~A1$<$A2}.    

%\vspace{0.05cm}

\item  \label{enc:rule16} \texttt{err\_t(I)~:-~lit(I,(pos,ba,F)), not true(F).}

\item   \label{enc:rule17} \texttt{err\_t(I)~:-~lit(I,(neg,ba,F)),~true(F).}

\item    \label{enc:rule18} \texttt{err\_d(I)~:-~lit(I,(pos,ac,(P,A))), not does(P,A).}

\item    \label{enc:rule19} \texttt{err\_d(I)~:-~lit(I,(neg,ac,(P,A))), does(P,A).}

\item   \label{enc:rule20} \texttt{legal(P,A)~:-~ha(I,(ac,(P,A))),~not err\_t(I).}

\item    \label{enc:rule21} \texttt{next(F)~:-~ha(I,(ba,F)), not err\_t(I), not err\_d(I).}
%\vspace{0.05cm}
\item  \label{enc:optimization}
            \texttt{:$\sim$~tup(I,TP,LIT). [cost(I,(TP,$\tau^{-1}$(LIT))),I,TP,LIT]}
      
\end{enumerate}
\caption{The ASP encoding of GDL repair, symbols like $+$,$-$,$\emptyset$ should be replaced by constants in actual implementation}
\label{fig:program}
\end{figure*}

%% file: experiment-2.tex
\section{Case Study}
\label{sec:exp}
%\yifan{Decide to change to a tic-tac-toe case study, the original one is too weird}
The motivation for introducing the GDL repair problem is that, in practice, game descriptions might violate the \yifanm{intention of the game designers. Once these violated properties are expressed in GTL and a cost function specified manually,} we can use our encoding in Section 4 to automatically generate minimal changes to fix a $\gd$. Our encoding works for \textbf{any} GDL repair task. We demonstrate this through a simple case study on Tic-Tac-Toe, beginning with the well-formedness property and then others.
\subsubsection{Instance Description} In Tic-Tac-Toe, two players $x$ and~$o$ alternate in marking cells on the board, beginning with~$x$. Turn-taking is modeled by the following pair of GDL rules\/: 
\begin{enumerate}
    \item $next(control(o)) \texttt{:-} true(control(x)).$
    \item $next(control(x)) \texttt{:-} true(control(o)).$
\end{enumerate}
The player who does not have control in a state can only perform a $noop$ action with no effect. The player who first places three of their marks in a horizontal, vertical, or diagonal row is the winner. If all 9 cells of the board are marked and neither $x$ nor $o$ wins, the game ends in a draw.

We purposefully break the $\gd$ by removing the second rule. As a result, player $x$ can never take control after step~1, and $o$ cannot take control after step~2. Consequently, the broken $\gd$ is not well-formed as it violates both the weak winnability and termination properties.
\subsubsection{Experimental Setup} 
To repair the broken $\gd$, we need to formulate an MRP of the form $\langle G, \Phi^{+},\Phi^{-},cost \rangle$ where $G=G_{L} \cup G_{N} \cup G_{R} \cup G_{E}$ (cf. Section~\ref{sec:problem-def}). We need to specify: $|G_{E}|$, the maximum number of new legal/next rules we can add; $\Phi^{+}$, the GTL properties that the repaired $\gd$ should satisfy; $\Phi^{-}$, the GTL properties that the repaired $\gd$ should dissatisfy; and a cost function. 

For simplicity, we set $|G_{E}|=2$ and use the following uniform cost function for all our experiments:
\begin{itemize}
    \item $cost(i,(+,l))=cost(i,(-,l))=1$
    \item $cost(i,(c,h))=1$, if $r_{i} \in G_{E}$
    \item $cost(i,(c,\emptyset))=|bd(r_{i})|+1$
    \item $cost(i,(c,h))=2 \cdot |bd(r_{i})|+2$, if $h \neq \emptyset$ and $r_{i} \in G_{o}$
\end{itemize}
This models the ``editing'' cost. Adding or deleting a body literal costs~1, as does adding the head of a new legal/next rule. We model deleting a rule (i.e., setting the head to~$\emptyset$) as deleting the head and all literals in the \emph{old} body of a rule. Likewise, replacing the head of an old rule is considered a significant change, costed as\/: removing the old rule and creating a new rule with a new head and the same body. 

All MRPs are solved by G\&C with Clingo~5.7.2 with the inverse linear search-based core-guided optimization configuration~\cite{lifschitz2019answer} on a Latitude 5430 laptop~\footnote{Source code link: \url{https://github.com/hharryyf/gdlRepair}}.

\subsubsection{Repairing the well-formedness property}
The most fundemantal, general property of any ``good'' $\gd$ is well-formedness, i.e., playability, termination, and weak winnability. We can repair a $\gd$ to be $n$-well-formed, for a give user-specified $n$ as the desired maximal horizon for the repaired game. In practice, $n$ would be based on domain knowledge of the game that the $\gd$ was intended to describe. For Tic-Tac-Toe, the desired maximal horizon is obviously~9, the number of cells that can be marked. Hence, in the MRP, we specify $\Phi^{+}=\{\psi_{play}(9),\psi_{end}(9)\}$ and $\Phi^{-}=\{\psi_{loss}(x,9),\psi_{loss}(o,9)\}$ (cf. Example~\ref{example:repair}). 

Our encoding from Section~\ref{sec:encoding} is able to automatically solve this MRP with an optimal repair at a cost of~1\/: The solution generated by our G\&C program is to simply create a new rule with an empty body: $next(control(x))$.
The resulting $\gd$ is well-formed, %making it valid for a GGP competition,
and it is syntactically close to the input, requiring only a single modification to the input.

\subsubsection{Refining the repair with additional constraints}
While the solution to our example Tic-Tac-Toe MRP successfully restores well-formedness, this may not be the only intended property for a game.
Specifically, after repairing the $\gd$ with the creation of the simple fact $next(control(x))$, the base proposition (aka.\ \emph{fluent\/}) $control(x)$ is true in all positions of the game. As a result, the repaired $\gd$ allows player~$x$ to always take control after step 1, which changes Tic-Tac-Toe from a turn-taking game to a simultaneous-move game. 

In a GDL description, having a next rule with an empty body is usually undesired as it allows a base proposition to persist unconditionally across all play sequences after step~1. Such a behavior can reduce the dimension of the state space of a game and hence oversimplify the game. 

\yifanm{If a game designer wants to avoid this behavior}, the following GTL formula, a kind of ``fluent dynamic constraint'' can be added to $\Phi^{-}$\/:
\begin{equation}
 \psi_{s}(f,n)\ \DefMath\ \neg \gdl{terminal} \wedge \X nest(true(f), \wedge,n-1)
    \nonumber
\end{equation}
This ensures that in the repaired $\gd$, either the game terminates at $\Sinit$, or there exists some $n$-max sequences where $true(f)$ does not hold continuously after step 1 which effectively rules out repairs that introduce rules like $next(f).$

Refining the MRP in this way means to find a repair with $\Phi^{-}=\{\psi_{loss}(x,9),\psi_{loss}(o,9),\psi_{s}(control(x),9)\}$ and $\Phi^{+}=\{\psi_{play}(9),\psi_{end}(9)\}$.

Now, the optimal repair generated by our G\&C has a cost of~2 to satisfy both well-formness and the additional fluent dynamic constraint. One repair output by G\&C suggests adding the rule: $next(control(x))\texttt{:-}\,true(control(o))$, which successfully restores the $\gd$ to the standard GDL description of Tic-Tac-Toe.

\subsubsection{Further refine the repair with the turn-taking constraint}
While the above minimal repair is arguably the desired one, the automatic solver G\&C outputs another lowest-cost repair of cost 2, namely, adding the rule: $next(control(x))\texttt{:-}\,not~does(x,mark(1,1)).$

The new rule says that $x$ will take control of the game in the next step as long as that player does not mark the cell $(1,1)$ in the current round. This is a perfectly \yifanm{appropriate} repair as the resulting variant of Tic-Tac-Toe is well-formed and does satisfy the additional constraint from above. However, this repair may still be considered undesirable because if~$x$ begins with marking the cell $(1,2)$ in step~1, for example, the player keeps control in step~2. In this case both $control(x)$ and $control(o)$ will be true in step~2, and hence they can still mark cells simultaneously in some game states.

If it is desirable to ensure that an $n$-well-formed two-player game after repairing satisfies a strict turn-taking property, which is to say that, in our example, at each step of the game either $x$ or $o$ take control of the game but not both, we can introduce the following additional GTL constraint\/:
\begin{equation}
        \psi_{con}(n) \ \DefMath\ \ nest(\psi_{x} \vee \psi_{o}, \wedge, n)
        \nonumber    
\end{equation}
where, $\psi_{x}=true(control(x)) \wedge \neg true(control(o))$ and $\psi_{o}=true(control(o)) \wedge \neg true(control(x))$.

To repair Tic-Tac-Toe with this extra constraint, the MRP uses: $\Phi^{-}=\{\psi_{loss}(x,9),\psi_{loss}(o,9),\psi_{s}(control(x),9)\}$ and $\Phi^{+} =\{\psi_{con}(9),\psi_{play}(9), \psi_{end}(9)\}$. 

Now the optimal repair still has cost~2, but the undesired repair suggested by G\&C when only considering the constraint to rule out a static fluent~$control(x)$ is eliminated. One automatically generated repair suggests to add the rule: $\gdl{next}(control(x)) \texttt{:-}\,\gdl{does}(x,noop)$.
This new rule states that whenever $x$ does $noop$, it will take control in the next step. In Tic-Tac-Toe, $x$ can only do $noop$ when $o$ takes control, which means this new rule is effectively ``equivalent'' to the original rule that we deleted---namely, $x$ taking control in the next step if $o$ takes control in the current step. 

More importantly, after introducing the turn-taking constraint, all lowest-cost repairs computed by G\&C result in a $\gd$ \emph{equivalent} to the originally correct Tic-Tac-Toe game description, in the sense that at every state of the repaired game and the original Tic-Tac-Toe $\gd$, the same set of legal actions are available to the players and each joint action leads to the same successor state in all these repaired games. 
\subsubsection{Summary}
The following table summarizes the lowest repair cost ($C_{opt}$) for each of the 3 MRPs that we considered in our study\/: repair the well-formedness property only (WF), repair the well-formedness property with the fluent dynamic constraint (WF + FD), and repair all 3 properties simultaneously (WF + FD + TT). We also record %the percentage of the optimal repair output by G\&C that restores the broken $\gd$ to an equivalent correct Tic-Tac-Toe $\gd$ (\%), 
the time for G\&C to find the first optimal repair ($T_{1}$), and the time to compute all optimal repairs ($T_{a}$).
\begin{center}
\footnotesize
\begin{tabular}{lc|rr}
\toprule
   \textbf{Property} & $C_{opt}$  & $T_{1}$ \textbf{(sec)} & $T_{a}$ \textbf{(sec)}  \\
  \hline
    WF & 1 &   34.76& 38.72 \\
    WF + FD &  2 & 50.99 & 326.41 \\
    WF + FD + TT & 2 & 122.29& 195.15\\
    \bottomrule
\end{tabular}    
\end{center}
% \begin{center}
% \footnotesize
% \begin{tabular}{lcr|rr}
%   \toprule
%    \textbf{Property} & $C_{opt}$  & \% & $T_{1}$ \textbf{(sec)} & $T_{a}$ \textbf{(sec)}  \\
%   \hline
%     WF & 1 & 0&  34.76& 38.72 \\
%     WF + FD &  2 & 18& 50.99 & 326.41 \\
%     WF + FD + TT & 2 & 100 & 122.29& 195.15\\
%     \bottomrule
% \end{tabular}    
% \end{center}
We observe that G\&C can find an appropriate repair for Tic-Tac-Toe in all 3 situations in a feasible amount of time with the resulting $\gd$ syntactically close to the original. Moreover, the case study demonstrates the ability of our encoding to automatically repair ill-defined GDL descriptions. \yifanm{Upon detecting---either manually or using a GTL theorem prover \cite{thielscher:AIJ12}---that a given $\gd$  violates human intentions, one can specify the desired properties as GTL formulas and use our approach to automatically generate a repaired $\gd$ that satisfies the intended properties while remaining syntactically closest to the original $\gd$. And, if the repair suggested by G\&C is still unsatisfactory, one can refine the repair by defining new MRPs with additional GTL properties based on \emph{additional} human knowledge of the resulting $\gd$, until the generated repair is satisfactory.}

% Note that while our approach is sound and complete for solving any valid MRPs, deciding the set of desired GTL properties of a $\gd$ is based on the intention of the game designer. There are existing works on automatically generating constraints from example traces is an interesting problem that has been studied in the context of LTL. Since GTL is a fragment of LTL, those techniques may be applied alongside our approach.

%And, if the repair suggested by G\&C is still unsatisfactory, one can refine the repair by defining new MRPs with additional GTL properties that the $\gd$ should satisfy or dissatisfy, until they are satisfied with the quality of the repair.  

%% file: conclusion.tex
\section{Conclusion}
\label{sec:conclusion}
% \munyquem{
% This paper investigates the problem of 
% repairing incorrect descriptions in Game Description Logic, with a focus on minimal repairs.
% We demonstrate sufficient conditions to ensure that certain repair problems have or do not have solutions. We proved tight complexity results for different computational problems related to minimal repairs. 
% Moreover, we provide the first automated method for repairing GDL descriptions based on  Answer Set Programming.

% For future work, we will explore the repair of games using strategic requirements specified in Alternating-time Temporal Logic \cite{alur2002alternating} or Strategy Logic \cite{mogavero2014reasoning}. Additionally, extending our approach to games with imperfect information described in GDL-II \cite{thielscher2010general} will allow the repair of games with epistemic requirements (e.g., the players should not see each other's cards).
% }
We investigated the problem of repairing GDL descriptions, with a focus on minimal repairs. We established sufficient conditions under which certain repair problems have, or do not have, solutions. We proved tight complexity bounds for the minimal repair problem and introduced the first automated method for repairing GDL descriptions using ASP, thereby extending the capabilities of automated theorem proving in GGP from mere fault detection to actual \emph{rectification}. \yifanm{One potential limitation of our automated method concerns efficiency and scalability. This is due to the complexity of the problem, but may be improved with the development of more efficient disjunctive ASP solvers.} %\munyque{Changed "Theoretical difficulty" with "complexity" or "computational complexity."}

\yifanm{For future work, we plan to design a broken GDL description dataset from existing descriptions~\cite{ggprepo} and systematically evaluate the performance of our encoding. We also intend to explore whether certain fragments of the repair problem (e.g., the $F\Delta_2^{P}$ fragment in Theorem~\ref{mrp:cheap}) can be solved more efficiently. 
Another future direction is to explore the repair of game properties formulated in more expressive logics, such as LTLf~\cite{bansal2023model}, as well as to investigate how our approach can be extended to games with imperfect information~\cite{thielscher2010general}, enabling the repair of epistemic properties~\cite{haufe2012automated}.}
%Another future direction is to explore repairing game properties formulated in more expressive logics, such as the full LTL~\cite{pnueli1977temporal}; or exploring how to extend our approach to games with imperfect information~\cite{thielscher2010general}, enabling the repair of epistemic properties~\cite{haufe2012automated}.}

%which supports temporal operators like \emph{eventually} and \emph{always}, allowing well-formedness properties to be expressed without relying on human knowledge of the horizon of the game.}

%to express well-formedness without the need for human knowledge on the horizon of the game.}

%\yifanm{TODO: rewrite future work}

% For future work, we aim to explore the repair of game properties formulated in more expressive logics, such as Alternating-time Temporal Logic \cite{alur2002alternating} and Strategy Logic \cite{mogavero2014reasoning}. We also plan to extend our approach to games with imperfect information \cite{thielscher2010general}, enabling the repair of epistemic properties.

\section*{Acknowledgements}

%\michaelm{I suggest we thank the reviewers, they've been quite remarkable. Anybody else?}\munyque{I agree, notice that the acknowledgments don't count towards the page limit, so we don't need to worry with space. }

\michaelm{We sincerely thank the anonymous reviewers for their insightful and thorough feedback, and for the care they invested in reviewing our work.}
%\yifanm{We thank the anonymous reviewers for their thoughtful and supportive comments.}
\munyquem{This project has received funding from the  European Union’s H2020 Marie Sklodowska-Curie project with grant agreement No 101105549. }

%% file: appendix.tex
\section{Proofs of Theorems}
\setcounter{theorem}{1}
\renewcommand{\thelemma}{\Alph{section}\arabic{theorem}}
\begin{theorem}
Let $G$ be a $\gd$ and $R$ be the set of players in $G$. For any $n>0$, there exists a repair $\mathcal{R}$ on $G$ such that $rep(G,\mathcal{R})$ is n-well-formed if \textbf{all} of the following hold:
    \begin{itemize}
        \item $G \cup \Sinit^{true} \not\models \gdl{terminal}$, and for all $p_{i} \in R$, $|\gamma(p_{i})| \geq 2$.
        \item For each $p_{i} \in R$, there exists a state $S_{i} \subseteq \beta$ such that $G \cup S^{true}_{i} \models \gdl{goal}(p_{i},100) \wedge \gdl{terminal}$.
        \item $|G_{L}| + |G_{E}| \geq 2 \cdot |R|$ and $|G_{N}| + |G_{E}| \geq |\mathcal{N}| \cdot |R|$ as well as $|G_{L}| + |G_{N}| + |G_{E}| \geq (2+|\mathcal{N}|) \cdot |R|$.
    \end{itemize}
\end{theorem}

\begin{proof}
    % Let $\mathcal{T}=\langle G,\Phi^{+},\Phi^{-},cost\rangle$ with $\gd$~$G$, $R$ be the set of players in $G$, $\Phi^{-}=\{\psi_{loss}(p,N) \mid p \in R\}$ and $\Phi^{+}=\{\psi_{end}(N),\psi_{play}(N)\}$ for some $N>0$. We've discussed in Example~\ref{example:repair} that the repair task has a solution iff $G$ can be repaired to be well-formed and with a maximum possible horizon $N$.
    % Hence, we only need to show that if all 3 conditions hold, $\mathcal{T}$ must have a solution for some $N$. We claim that if all 3 conditions hold, $\mathcal{T}$ with any $N=1$ must have a solution repair. Note that this also implies that $\mathcal{T}$ with any $N > 0$ must have a solution repair. In other words, as long as we can show that there is some repair that can transform $G$ into a game with horizon 1 such that it is weakly winnable by all players, termination after 1 step, and playable at $S_{0}$, we are done.
    We show that $G$ can always be repaired to be 1-well-formed. Note that if a game is 1-well-formed, it is definitely n-well-formed for any $n>0$.
    
    Let $p_{1},\ldots,p_{|R|}$ be the names of $|R|$ players.
    Note that any valid GDL in restricted form with at most $|G_{L}|+|G_{E}|$ (resp. $|G_{N}|+|G_{E}|$) legal (resp. next) rules, and at most $|G_{C}|$ legal or next rules that has the same base, move domain, and $G_{R}$ can be obtained from the $G$ with some valid repair.

    Since any GDL with no more than $2 \cdot |R|$ legal rules and no more than $|\mathcal{N}| \cdot |R|$ next rules satisfies condition 3, as long as we can construct a game $G'$ with no more than $2 \cdot |R|$ legal rules and no more than $|\mathcal{N}| \cdot |R|$ next rules such that $G'$ is 1-well-formed, we can ensure that there must exist a repair $\mathcal{R}$ such that $G'=rep(G,\mathcal{R})$, and the theorem holds. 

    From condition 1, $|\gamma(p_{i})| \geq 2$ for all $p_{i} \in R$, we know that we can pick two actions $a_{i}^{1}$ and $a_{i}^{2}$ from $\gamma(p_{i})$. From condition 2, we know that we can pick $|R|$ states where the $i$-th one is $S_{i}$ such that $S_{i} \subseteq \beta$ and $G \cup S^{true}_{i} \models \gdl{goal}(p_{i},100) \wedge \gdl{terminal}$.
    
    The GDL construction is as follows. It uses exactly $2 \cdot |R|$ legal rules and no more than $|\mathcal{N}| \cdot |R|$ next rules.

    For the $i$-th player $p_{i}$, construct 2 legal rules.
    \begin{equation}
    \label{eq:lg1}
        legal(p_{i},a_{i}^{1}). \tag{E2}
    \end{equation}
    \begin{equation}
    \label{eq:lg2}
        legal(p_{i},a_{i}^{2}). \tag{E3}
    \end{equation}
    For each player $p_{i}$ and each $f \in S_{i}$ where $i < |R|$, construct the following next rule.
    \begin{multline}
    \label{eq:nxt1}
        next(f) \texttt{:-} does(p_{1},a_{1}^{2}), \ldots,does(p_{i-1},a_{i-1}^{2}),\\does(p_{i},a_{i}^{1}). \tag{E4}
    \end{multline}
    And for each $f \in S_{|R|}$, construct the following next rule.
    \begin{equation}
    \label{eq:nxt2}
        next(f) \texttt{:-} does(p_{1},a_{1}^{2}), \ldots,does(p_{|R|-1},a_{|R|-1}^{2}). \tag{E5}    
    \end{equation}
    Observe that we created 2 legal rules for each player. And since $S_{i} \subseteq \beta$, we ensure that the total number of next rules is no more than $|R| \cdot |\mathcal{N}|$.

    In the construction, ~(\eqref{eq:lg1}) and~(\eqref{eq:lg2}) ensure that the game is playable at all states, and at each state, there are $2^{|R|}$ possible joint actions. Condition 1 ensures that the game cannot terminate at the initial state. 
    ~(\eqref{eq:nxt1}) and~(\eqref{eq:nxt2}) ensure that each of the $2^{|R|}$ possible legal joint actions at $\Sinit$ can lead to one of the next state within $S_{1},\ldots,S_{|R|}$. Concretly, $2^{|R|-1}$ joint actions lead to $S_{1}$, $2^{|R|-2}$ joint actions lead to $S_{2},\ldots$, 2 joint actions lead to $S_{|R|-1}$ and 2 joint actions lead to $S_{|R|}$.

    Since all $S_{i}$ are terminal states and $S_{i}$ is winning state for $p_{i}$, reachable with some legal joint actions from $\Sinit$, we know that the constructed game has a horizon 1, and is weakly winnable and playable. This confirms that if all 3 conditions are satisfied, $G$ can definitely be repaired to be 1-well-formed, hence, n-well-formed for any $n>0$.
\end{proof}
\begin{theorem}
Let $\mathcal{T}\!=\!\langle G,\Phi^{+},\Phi^{-},cost\rangle$ be a repair task with $\psi_{end}(n) \!\in\! \Phi^{+}$ for some $n$. Let $K\!=\!max(1,|\Phi^{-}|)$ and $R$  the set of players in $G$.
If $\mathcal{T}$ has no solution and $|G_{E}|\!\geq\! K \cdot (n+1) \cdot |\mathcal{L}|+ n \cdot K\cdot |\mathcal{N}| \cdot (|R| + 1)$, then any repair task $\mathcal{T'}\!=\!\langle G',\Phi^{+},\Phi^{-},cost'\rangle$, where $G'\!=\!G_{L} \cup G_{N} \cup G_{E}' \cup G_{R}$ (with $|G_{E}'|$ empty rules), also has no solution. 
\end{theorem}  
\begin{proof}
    Recall that $\beta$ is the set of base propositions, $\gamma$ is the move domain, and $\gamma(p)$ is the move domain of player $p$.
    
    First, it is trivial that if a repair task has solutions, then increasing the size of $|G_{E}|$ won't make it unsolvable, because we can simply ignore those additional empty rules and use the old solution repair as the solution to the new problem. 

    Hence, the theorem is equivalent to showing that if the repair task $T'=\langle G',\Phi^{+},\Phi^{-},cost' \rangle$ where $G'=G_{L} \cup G_{N} \cup G_{E}' \cup G_{R}$ has a solution, the repair task $T''=\langle G'',\Phi^{+},\Phi^{-},cost'' \rangle$ where $G''=G_{L} \cup G_{N} \cup G_{E}'' \cup G_{R}$ and $|G_{E}''|=K \cdot (n+1) \cdot |\mathcal{L}|+ n \cdot K\cdot |\mathcal{N}| \cdot (|R| + 1)$ must have a solution. 
     In other words, as long as we can construct a $\gd$ $G^{sol}$ with at most $K \cdot (n+1) \cdot |\mathcal{L}|+ n \cdot K\cdot |\mathcal{N}| \cdot (|R| + 1)$ legal/next rules in total such that $G^{sol} \models_{t} \phi^{+}$ for all $\phi^{+} \in \Phi^{+}$ and $G^{sol} \not\models_{t} \phi^{-}$ for all $\phi^{-} \in \Phi^{-}$ based on the fact that $\mathcal{T}'$ has a solution, we can repair $G''$ by eliminating all rules in $G''_{L}$ and $G''_{N}$ and replace the section $G''_{E}$ with $G^{sol}_{L} \cup G^{sol}_{N}$, and hence we find a solution to $\mathcal{T''}$ and we are done.

    Let's construct $G^{sol}$ based on the fact that $\mathcal{T}'$ has a solution repair $\mathcal{R}$. Let $G^{1}=rep(G',\mathcal{R})$, we know that $G^{1} \models_{t} \phi^{+}_{i}$ for all $\phi^{+}_{i} \in \Phi^{+}$ and $G^{1} \not\models_{t} \phi^{-}_{i}$ for all $\phi^{-}_{i}~\in~\Phi^{-}$.   
    Suppose that $\phi^{+}=\phi^{+}_{1} \wedge \phi^{+}_{2} \ldots \wedge \phi^{+}_{|\Phi^{+}|}$, $G^{1} \models_{t} \phi^{+}$ must hold. In addition, since $\psi_{end}(n) \in \Phi^{+}$ and $deg(\psi_{end}(n))=n$, we have that $deg(\phi^{+}) \geq n$ and there exists no sequences $(S_{0},\myDots,S_{m})$ such that $m = n$ and $S_{m} \notin B \cup T$ where $T$ is the set of all terminal states, and $B$ is the set of all non-playable states in the game described by $G^{1}$. Furthermore, since $G^{1}$ is a solution $\gd$, there must exists $K$ play sequences $(S_{0}^{i},\myDots,S_{d_{i}}^{i})$ where $S_{0}^{i}=\Sinit$ for all $1 \leq i \leq K$ such that the $i$-th one ensures that $G^{1},(S_{0}^{i},\myDots,S_{d_{i}}^{i}) \not\models_{t} \phi^{-}_{i}$ and $G^{1},(S_{0}^{i},\myDots,S_{d_{i}}^{i}) \models_{t} \phi^{+}$. We also know that $0 \leq d_{i} \leq n$ for all $1 \leq i \leq K$, and $S^{i}_{d_{i}} \in T \cup B$ must hold. Note that for the case when $\Phi^{-}=\{\}$, we have that $K=1$, and there must exist a sequence $(S_{0}^{1},\myDots,S_{d_{1}}^{1})$ with $G^{1},(S_{0}^{1},\myDots,S_{d_{1}}^{1}) \models_{t} \phi^{+}$. 
    
    \subsubsection{Auxiliary notations} Before we proceed with the proof, we define some auxiliary notations. 
    
    We define the set of end states $\mathcal{E}$ as:
    \begin{equation*}
        \mathcal{E}=\{S^{i}_{d_{i}} \mid 1 \leq i \leq K\}
    \end{equation*}

    We define the set of states $\mathcal{C}$ as:
    \begin{equation*}
        \mathcal{C} = \bigcup_{i=1..K} \{S \mid S=S^{i}_{j}~and~0 \leq j \leq d_{i}\}
    \end{equation*}
    Intuitively, $\mathcal{C}$ is the set of states $S \subseteq \beta$ that are on at least one of the $K$ play sequences we discussed above.

    For any 2 states $S_{1},S_{2} \subseteq \beta$, we write $S_{1} \Rightarrow S_{2}$ iff 
    \begin{itemize}
        \item $S_{1} \in \mathcal{C}$ and $S_{2} \in \mathcal{C}$, and 
        \item For some $0 \leq j \leq d_{i} -1$ and $1 \leq i \leq K$,  $S_{1}=S^{i}_{j}$ and $S_{2}=S^{i}_{j+1}$. 
    \end{itemize}
    Intuitively, $S_{1} \Rightarrow S_{2}$ means $S_{2}$ is the successor of $S_{1}$ on \emph{some} of the $K$ playing sequences. 

    For any $S \in \mathcal{C}$ we define 
    \begin{itemize}
        \item $succ(S)=\{S' \in \mathcal{C} \mid S \Rightarrow S'\}$.
    \end{itemize}
    which is the set of all successor states of $S$. 
    
    For any $S \in \mathcal{C}$ and $p \in R$ we define
    \begin{itemize}
        \item $lg(S,p)=\{a\in \gamma(p) \mid G^{1} \cup S^{true} \models legal(p,a)\}$
    \end{itemize}
    which is the set of all legal actions for $p$ at state $S$ in $G^{1}$.

    We continue with the proof by first considering the following statement.
    
    \subsubsection{Statement 1.} For any play sequences $(S_{0},\myDots,S_{m})$ of $G^{1}$, if $S_{i} \Rightarrow S_{i+1}$ for all $0 \leq i \leq m-1$ and $S_{m} \in T \cup B$, then $G^{1},(S_{0},\myDots,S_{m}) \models_{t} \phi^{+}$.
    
    The statement trivially holds because $(S_{0},\myDots,S_{m})$ is just one of the ending sequences in $G^{1}$ and $G^{1}$ is a solution $\gd$. 

    \subsubsection{Statement 2.} If $G^{1}$ is a solution $\gd$, then there must exist a $\gd$ $G^{sol}$ that satisfies all the following 5 conditions, and $G^{sol} \models_{t} \phi^{+}$ and $G^{sol} \not\models_{t} \phi^{-}_{i}$ for all $\phi^{-}_{i} \in \Phi^{-}$. 
    \begin{enumerate}
        \item The total number of legal/next rules in $G^{sol}$ is no more than $K \cdot (n+1) \cdot |\mathcal{L}|+ n \cdot K\cdot |\mathcal{N}| \cdot (|R| + 1)$.
        \item For every ground atom $q(\vec t\,)$ that is an instance of some predicate $q \in G^{1}$ that is not $\gdl{init}$ and does not depend on $\gdl{does}$, and for every $S \in \mathcal{C}$, we have: $G^{sol} \cup S^{true} \models q(\vec t\,)$ iff $G^{1} \cup S^{true}\models q(\vec t\,)$
        \item $\mathcal{C}$ is the set of all states that are reachable from $S_{0}$ in $G^{sol}$. 
        \item For every state $S \in \mathcal{C} \setminus \mathcal{E}$ and every $S' \in succ(S)$, there exists a joint action $A$ (cf. Sec.~\ref{sec:prelim}) such that $A(p) \in lg(S,p)$ for every $p \in R$ and $S'=\{f \in \beta \mid G^{1} \cup S^{true} \cup A^{does} \models next(f) \}$. 
        \item For every state $S \in \mathcal{C} \setminus \mathcal{E}$ and joint action $A$ such that $A(p) \in lg(S,p)$ for every $p \in R$. If $S'=\{f \in \beta \mid G^{1} \cup S^{true} \cup A^{does} \models next(f) \}$, then $S' \in succ(S)$.
    \end{enumerate}

    Condition 1 ensures that the total number of rules in $G^{sol}$ is within the upper bound; hence, it is obvious that if we can prove statement 2, we prove the entire theorem. 
    
    We claim that if there exists a $\gd$ $G^{sol}$ that satisfies all 5 conditions, then $G^{sol} \models_{t} \phi^{+}$ and $G^{sol} \not\models_{t} \phi^{-}_{i}$ for all $\phi^{-}_{i} \in \Phi^{-}$ is automatically satisfied. 

    Condition 3 ensures that the state space of $G^{sol}$ is a subset of the state space of $G^{1}$. Condition 2 ensures that at any state $S$ in $G^{sol}$, any ground atom that can be part of any GTL expression holds if and only if it holds at state $S$ in $G^{1}$. Conditions 3 and 4 ensure that for every $1 \leq i \leq K$, $(S_{0}^{i},\myDots,S_{d_{i}}^{i})$ is still a valid sequence in $G^{sol}$. Therefore, combined with condition 2, $G^{sol} \not\models_{t} \phi_{i}^{-}$ still holds for all $\phi_{i}^{-} \in \Phi^{-}$. Condition 5 ensure that all sequences $(S_{0},\myDots,S_{k})$ in $G^{sol}$ must have $S_{0},S_{1},\ldots,S_{k} \in \mathcal{C}$ and $S_{i} \Rightarrow S_{i+1}$ for every $0 \leq i < k$. Hence, by Statement 1, we can ensure that $G^{sol} \models_{t} \phi^{+}$. This implies that $G^{sol}$ is a solution $\gd$ as long as $G^{1}$ is. Therefore, our proof reduces to constructing a $\gd$ $G^{sol}$ in restricted form that satisfies all the 5 conditions above based on the fact that \emph{there exists} some solution $\gd$ $G^{1}$. 

    Since $G_{R}$ cannot be changed by the repair operations (i.e., $G_{R}=G_{R}^{1}=G_{R}^{sol}$) and $G^{1}$ is in restricted form, we know that for every predicate symbol $q \neq legal$, condition 2 automatically holds in $G^{sol}$. Hence, the first step of our construction of $G^{sol}$ is to satisfy condition~2 for the case when $q=legal$. To do so, we create the following legal rule for each $S \in \mathcal{C}$, $p \in R$, and each $a \in lg(S,p)$:
    \begin{equation}
        \label{eq:legalrules}
        legal(p,a)\texttt{:-}\bigwedge_{f \in S} true(f),\bigwedge_{g \in \beta \setminus S} not~true(g). \tag{E6}
    \end{equation}
    Here, $\bigwedge_{f \in S} true(f)$ abbreviates $true(f_{1}),\ldots,true(f_{k})$ for all $f_{1},\ldots,f_{k} \in S$, and $\bigwedge_{f \in S'} not~true(f)$ abbriviates $not~true(f_{1}),\ldots,not~true(f_{t})$ for all $f_{1},\ldots,f_{t} \in S'$.

    To satisfy conditions 3, 4 and 5, for each state $S \in \mathcal{C} \setminus \mathcal{E}$, we need to map all joint actions $A$ (cf. Section 2) such that $A(p) \in lg(S,p)$ for all $p \in R$ to some $S'$ with $S \Rightarrow S'$ such that $S~\xrightarrow[]{A} S'$. Additionally, for each $S' \in succ(S)$, there must be some legal joint actions $A$ such that $S~\xrightarrow[]{A} S'$.

    Our strategy is to first link each successor state $S'$ of $S$ with some legal joint actions $A$, and for the remaining legal joint actions, we map all of them to a ``default successor'' of $S$. For the default successor of $S$, we denote it as $S_{d}$, which can be any element in $succ(S)$.

    For each state $S \in \mathcal{C} \setminus \mathcal{E}$, and for each of the $|R|$ players, we denote the legal actions of player $i$ as $a_{i}^{0},a_{i}^{1},\ldots,a_{i}^{l_{i}-1}$ where $l_{i}=|lg(S,p_{i})|$.
    We define a joint action vector as $\langle b_{1},b_{2},\ldots,b_{|R|-1},b_{|R|} \rangle$, where $0 \leq b_{i} \leq l_{i}-1$ for all $1 \leq i \leq |R|$, which corresponds to the joint action $\{(p_{1},a_{1}^{b_{1}}),(p_{2},a_{2}^{b_{2}}),\ldots,(p_{|R|},a_{|R|}^{b_{|R|}})\}$. Apparently there are $V = l_{1} \cdot l_{2} \cdot \ldots \cdot l_{|R|}$ many joint action vectors at state $S$. Since $G^{1}$ is a solution $\gd$, we know that $V \geq |succ(S)|$.

    Consider two joint action vectors: $v_{1}=\langle b_{1},b_{2},\ldots,b_{|R|} \rangle $
   and $v_{2}=\langle b'_{1},b'_{2},\ldots,b_{|R|} \rangle $.
    We say $v_{1}$ is lexicographically smaller than $v_{2}$ iff there is some $1 \leq i \leq |R|$ such that $b_{i} < b'_{i}$ and $b_{j}=b'_{j}$ for all $1 \leq j < i$. Note that this is the same as the definition of one vector being smaller than the other in mathematics. For example, the action vector $\langle 1, 2, 4 \rangle $ is smaller than $\langle 1, 3, 2\rangle$. We denote the $i$-th smallest joint action vector at $S$ to be $v_{i}=\langle b^{i}_{1},b^{i}_{2},\ldots,b^{i}_{|R|} \rangle$. We first map $S$ to every successor $S' \in succ(S)$ with the joint actions that are lexicographically smaller than or equal to $v_{|succ(S)|}$.  

    Suppose that $succ(S)=\{S'_{1},\ldots,S'_{|succ(S)|}\}$, we create the following next rule for every $1 \leq i \leq |succ(S)|$ and every base proposition $g \in S'_{i}$.
    \begin{multline}
        \label{eq:nextrule1}
        next(g) \texttt{:-} \bigwedge_{f \in S} true(f), \bigwedge_{f \in \beta \setminus S} not~ true(f), \\does(p_{1},a^{b_{1}^{i}}_{1}),~does(p_{2},a^{b_{2}^{i}}_{2}),
        \\\ldots,does(p_{|R|},a^{b_{|R|}^{i}}_{|R|}). \tag{E7}
    \end{multline}
    Note that for each $S'_{i}$, we need at most $|\mathcal{N}|$ many such rules, this rule simply says that $S$ can go to $S'_{i}$ with the joint action corresponds to the joint action vector $v_{i}$.

    We then map all joint actions such that their joint action vector is greater than $v_{|succ(S)|}$ to a default successor of $S$, $S_{d}$. For clarity, we let $c=|succ(S)|$. 
    We create the following next rule for each $g \in S_{d}$ and $1 \leq i \leq |R|$.
    \begin{multline}
        \label{eq:nextrule2}
        next(g) \texttt{:-} \bigwedge_{f \in S} true(f), \bigwedge_{f \in \beta \setminus S} not~ true(f), \\does(p_{1},a^{b^{c}_{1}}_{1}),does(p_{2},a^{b^{c}_{2}}_{2}), \ldots,
does(p_{i-1},a^{b^{c}_{i-1}}_{i-1}), \\
not~does(p_{i},a^{1}_{i}),not~does(p_{i},a^{2}_{i}), \ldots,not~does(p_{i},a^{b^{c}_{i}}_{i}).
 \tag{E8}
    \end{multline}
    For each $i$, the rule of this form says that if a joint action vector $v'$ is greater than $v_{|succ(S)|}$ (aka. $v_{c}$), it can have the first $i-1$-th bit the same as $v_{c}$ (i.e., the second line) while the $i$-th bit is something greater than $b^{c}_{i}$ (i.e., the third line). If that's the case, the joint action should transit $S$ to the default successor $S_{d}$.

    The total number of legal rules we use in the form~(\eqref{eq:legalrules}) is at most the total number of states times the size of the overall move domain, which is $K \cdot (n+1) \cdot  |\mathcal{L}|$.

    The total number of next rules of the form~(\eqref{eq:nextrule1}) we use for connecting each state $S$ with its successors at least once is $\sum_{S \in \mathcal{C} \setminus \mathcal{E}} |succ(S)| \cdot |\mathcal{N}|$, which is $|\mathcal{N}| \cdot \sum_{S \in \mathcal{C} \setminus \mathcal{E}} |succ(S)|$. The total number of successor states of all states is bounded by $d_{1}+d_{2}+\ldots+d_{K}$. Note that since $\phi_{end}(n) \in \Phi^{+}$, we know that $d_{i} \leq n$ for all $i$. Thus, $\sum_{S \in \mathcal{C} \setminus \mathcal{E}} |succ(S)| \leq n \cdot K$.

    The total number of next rules of form~(\eqref{eq:nextrule2}) we use for connecting each state $S$ with that default successor $S_{d}$ is at most $|\mathcal{N}| \cdot |R|$ rules for each state $S \in \mathcal{C} \setminus \mathcal{E}$. In total, this is at most $n \cdot K \cdot |\mathcal{N}| \cdot |R|$.

    To sum up, the total number of legal/next rules we use is at most $K \cdot (n+1) \cdot  |\mathcal{L}|+n \cdot |\mathcal{N}| \cdot K \cdot (1+|R|)$, and condition 1 holds and so does Statement 2. 

    In conclusion, as long as there exists a solution $\gd$ $G^{1}$, we know that there must exist a solution $\gd$ $G^{sol}$ that contains at most $K \cdot (n+1) \cdot  |\mathcal{L}|+n \cdot |\mathcal{N}| \cdot K \cdot (1+|R|)$ legal/next rules and the proof is complete. 
    % use i, j 
    % eliminate condition 3, keep condition 4
    
%     For each $S \in \mathcal{S}$, and $(p,a) \in L(S)$ create the following legal rules:
%     \begin{equation}
%     \label{rule:legal}
%         legal(p,a):-\bigwedge_{f \in S} true(f),\bigwedge_{f \in \beta \setminus S} not~true(f). \tag{E6}
%     \end{equation}

    % Also, for any 2 states $S_{1},S_{2} \subseteq \beta$, we write $S_{1} \rightarrow S_{2}$ iff 
    % \begin{itemize}
    %     \item $S_{1} \in \mathcal{C}$ and $S_{2} \in \mathcal{C}$, and 
    %     \item There is a legal joint action $A$ such that $(p,A(p),S_{1}) \in l$ for each $p \in R$ and $u(A,S_{1})=S_{2}$ (cf. $\defshort$~\ref{def:gdl}).
    % \end{itemize}
    % Note that $S_{1} \rightarrow S_{2}$ simply means that $S_{2}$ can be a successor state of $S_{1}$ in the game defined by $G^{1}$ and both states must appear in $\mathcal{C}$. It is trivial that $S_{1} \Rightarrow S_{2}$ implies $S_{1} \rightarrow S_{2}$ but not vice versa.
\end{proof}

\setcounter{theorem}{4}
\renewcommand{\thelemma}{\Alph{section}\arabic{theorem}}
\begin{theorem}
    $\mrpt$ is $\Delta_{3}^{P}$-complete.
\end{theorem}
\begin{proof} (Membership) Suppose the change tuple is $\langle i, dom \rangle$. Since the maximum possible cost of a valid repair is bounded ($\defshort$~\ref{def:repset} and~\ref{def:cost}) and the answer to $\mrpb$ is monotonic, we can binary search for the lowest cost bound $C$ such that the $\mrpb$ returns ``yes''. If such a $C$ doesn't exist, the answer to $\mrpt$ is ``no''. Once we have the smallest $C$ after the binary search procedure, we can check if a change tuple is some lowest-cost repair with an additional $\mrpb$ oracle call. Basically, we set a new cost function: $cost'(j,dom')=2 \cdot cost(j,dom')$ if $\langle j,dom' \rangle \neq \langle i,dom \rangle$, and $cost'(i,dom) \!= \!2 \cdot cost(i,dom) \!- \! 1$. 
Here, $j \in \mathcal{I}$ and $dom \in Dom$ still holds. Then, we query $\mrpb$ with $\langle G,\Phi^{+},\Phi^{-},cost',2 \cdot C-1 \rangle$. We claim that the answer to $\mrpt$ is yes iff the modified $\mrpb$ returns true. Note that, under the new cost function, the cost of any repair that contains $\langle i,dom \rangle$ is $2 \cdot A - 1$ where $A$ is the cost of the repair under the old cost function. Similarly, the cost of any repair that does not contain $\langle i,dom \rangle$ is $2 \cdot A$ under the old cost function. If the new $\mrpb$ returns yes, it means that the optimal solution costs $2 \cdot C -1$ for the new MRP, thus, there is a repair cost of $C$ for the old repair task. Hence, $\mrpt$ returns yes. Conversely, if $\mrpt$ returns yes, the answer to $\mrpb$ must be yes, because there is a solution repair of cost $2 \cdot C - 1$ by using the same repair for the old MRP. %Hence, the optimal repair cost is $10 \cdot C - 1$ under the new cost function iff there is some repair that has a cost of $C$ under the old cost function that contains  $\langle i,dom \rangle$. 
The total number of $\mrpb$ oracle calls is polynomial w.r.t. the size of the input. By $\theoremshort$~\ref{brp:sigma2}, $\mrpt$ is in $P^{\Sigma_2^{P}}=\Delta_{3}^{P}$. 
    
    (Hardness) For a TQBF of the form~(\eqref{qbf}), we create an $\mrpt$ with the same $G$, $\Phi^{+}$, $\Phi^{-}$, and cost function as in $\theoremshort$~\ref{brp:sigma2}. Since~(\eqref{qbf}) is true, we know that the optimal repair costs $\leq 2^{n}\!-\!1$. Note that because of the way the cost function is constructed (i.e., $1,2,4,\ldots,2^{n-1}$), we know that the optimal repair is also unique. 
    The proof of $\theoremshort$~\ref{brp:sigma2} shows a one-to-one correspondence between the satisfiable assignments of existential variables in~(\ref{qbf}) and valid repairs with cost $\leq 2^{n}\!-\!1$. Concretely, for a TQBF of the form (\eqref{qbf}), let $\sigma = \existblock \rightarrow \{\top,\bot\}$ be an assignment to the variables in $\existblock$ of (\eqref{qbf}). We know that $\sigma$ is a satisfiable assignment to existential variables in $\Psi$ (i.e., $\forall \univblock ~E_{\sigma}$ is a tautology) iff the repair containing change tuples $\{\langle i+2m ,(-,true(x_{i}))\rangle~|~\sigma(x_{i})=\top\}$ is a solution repair to the $\mrpb$ defined in $\theoremshort$~\ref{brp:sigma2}. Hence, the $\Delta_{3}^{P}$-hard problem ``Deciding if the lexicographically minimum satisfiable assignment of the existential variables in a TQBF of form (\ref{qbf}) has $x_{n}=\top$~''~\cite{mazzotta2024quantifying,krentel1992generalizations} 
    can be reduced to checking if the change tuple $\langle n\!+\!2m,(-,true(x_{n})) \rangle$ is in the optimal repair $\mathcal{R}$. This implies that $\mrpt$ is $\Delta_3^{P}$-hard. 
\end{proof}
\begin{theorem}
    MRP is $F\Delta_3^{P}$-complete.
\end{theorem}
\begin{proof}
    (Membership) Consider an MRP instance $\langle G, \Phi^{+}, \Phi^{-}, cost \rangle$.
    Suppose the change tuple is $\langle i, dom \rangle$. Since the maximum possible cost of a valid repair is bounded ($\defshort$~\ref{def:repset} and~\ref{def:cost}) and the answer to $\mrpb$ is monotonic, we can binary search for the lowest cost bound $C$ such that the $\mrpb$ returns ``yes''. If such a $C$ doesn't exist, there is no solution to MRP. Once we have the smallest $C$ after the binary search procedure, we modify the cost function to be $cost'(i,dom)=2 \cdot cost(i,dom)$, for all $\langle i, dom\rangle \in \mathcal{I} \times Dom$.
    
    We iterate through all elements of $\langle i, dom \rangle \in \mathcal{I} \times Dom$. We define $\mathcal{R}=\{\}$, and let $n=1$. For each $\langle i, dom \rangle$, we do the following:
    \begin{enumerate}
        \item Let $cost'(i,dom)=2 \cdot cost(i,dom)-1$. Query the $\mrpb$ oracle for the problem $\langle G, \Phi^{+}, \Phi^{-}, cost', 2 \cdot C - n \rangle$.
        \item If the result is ``yes'', add $\langle i, dom \rangle$ to $\mathcal{R}$, and set $n=n+1$.
        \item Otherwise, set $cost'(i,dom)=2 \cdot cost(i,dom)$.
    \end{enumerate}
    $\mathcal{R}$ is the optimal solution repair. Since the total number of calls to the $\mrpb$ oracle is polynomial w.r.t. the size of the input, we know that MRP is in $F \Delta_{3}^{P}$.
    
    (Hardness) Similar to the proof of Theorem~\ref{mrp:tuple}. For a TQBF of the form~(\eqref{qbf}), we create an $\mrpt$ with the same $G$, $\Phi^{+}$, $\Phi^{-}$, and cost function as in $\theoremshort$~\ref{brp:sigma2}. We reduce the $F\Delta_{3}^{P}$-hard problem ``Finding the lexicographically minimum satisfiable assignment of the existential variables in a TQBF of form (\ref{qbf})''~\cite{krentel1992generalizations} to finding the optimal repair of the MRP. And we should know that the lexicographically minimum satisfiable assignment has $x_{i}=\top$ iff $\langle i+2m,(-,true(x_{i}))\rangle$ is in the optimal repair, which completes the hardness part of the proof.  
\end{proof}
\begin{theorem}
          $\mrpb$ is NP-complete, %$\mrpv$ $\text{co-NP}$-complete, 
     $\mrpt$ $\Delta_{2}^{P}$-complete, and MRP $F\Delta_{2}^{P}$-complete if $\Phi^{+}=\{\}$.      
\end{theorem}
\begin{proof}
        We only show that $\mrpb$ is NP-complete. Since proving $\mrpt$ is $\Delta_2^{P}$-complete, and MRP is $F\Delta_2^{P}$-complete when $\Phi^{+}=\{\}$ is very similar to the process of showing general $\mrpt$ is $\Delta_3^{P}$-complete, and MRP is $F\Delta_3^{P}$-complete after showing $\mrpb$ is $\Sigma_2^{P}$-complete.
        
        (Membership) One can guess a repair $\mathcal{R}$ and a $deg(\phi_{i})$-max sequence $(\Sinit^{i},\ldots,S^{i}_{n_{i}})$ where $\Sinit^{i}=\Sinit$ for each $1 \leq i \leq |\Phi^{-}|$. One can check if $\mathcal{R}$ costs $\leq C$ and $rep(G,\mathcal{R}),(\Sinit^{i},\ldots,S^{i}_{n_{i}}) \not\models_{t} \phi_{i}$  for each $\phi_{i} \in \Phi^{-}$ in PTIME. Hence, $\mrpb$ with $\Phi^{+}=\{\}$ is in NP. 
        
        (Hardness) Consider a SAT formula
        \begin{equation}
            \label{eq:sat}
            \Psi = C_{1} \wedge C_{2} \wedge \ldots \wedge C_{m} \tag{E10}
        \end{equation}
        such that $C_{i}=l_{i}^{1} \vee l_{i}^{2} \vee l_{i}^{3}$ where each $l_{i}^{j}$ is a literal over variables $X=\{x_{1},\ldots,x_{n}\}$. We know that deciding if $\Psi$ is satisfiable is NP-hard. We reduce such a problem to an $\mrpb$ with $\Phi^{+}=\{\}$. Consider a $\gd$ $G$ with $n$ players $p_{1},\ldots,p_{n}$, base propositions $c_{1},\ldots,c_{m}$, and the move domain of each player is $pos$ or $neg$. $G$ contains the following rules:
        \begin{itemize}
            \item $R_{1}=\bigcup_{i=1..n} \{[r_{i}]~ legal(p_{i},neg).\}$
            \item $R_{2}=\bigcup_{i=1..m} \{[r_{i+n}]~next(c_{i}) \texttt{:-} \mu(l_{i}^{1})\}$.
            \item $R_{3}=\bigcup_{i=1..m} \{[r_{i+n+m}]~next(c_{i}) \texttt{:-} \mu(l_{i}^{2})\}$.
            \item $R_{4}=\bigcup_{i=1..m} \{[r_{i+n+2m}]~next(c_{i}) \texttt{:-} \mu(l_{i}^{3})\}$.
            \item $R_{5}=goal(p_{1},100) \texttt{:-} true(c_{1}),\ldots,true(c_{m}).$
        \end{itemize}
         where, $\mu(l_{i}^{j})=\begin{cases} 
        does(p_{k},pos)& if~l_{i}^{j}=x_{k}~and~x_{k} \in X\\
        does(p_{k},neg)& if~l_{i}^{j}=\neg x_{k}~and~x_{k} \in X\\
    \end{cases}$
    Let $C=2^{n}-1$, $G_{L}=R_{1}$, $G_{N}=R_{2} \cup R_{3} \cup R_{4}$, $G_{R}=R_{5}$, $|G_{E}|=0$, $\Phi^{+}=\{\}$, and $\Phi^{-}=\{\neg \X goal(p_{1},100)\}$. 
    The cost function is defined as (recall that $G_{C}=G_{L} \cup G_{N} \cup G_{E}$ and $\mathcal{I}=\{1,2,\myDots,|G_{C}|\}$):
    \begin{itemize}
        \item $cost(i,(c,legal(p_{i},pos)))=2^{n-i}$ for $1 \leq i \leq n$
        \item $cost(i,dom)=2^{n}$, for other $i \in \mathcal{I}$ and $dom \in Dom$.
    \end{itemize}
    Informally, we have a GD that has rules of form $R_{1}$ to $R_{5}$, and changing the head of the $i$-th rule of form $R_{1}$ to $legal(p_{i},pos)$ costs $2^{n-i}$ (i.e., $2^{n-1},\ldots,2,1$) while all other modifications to the game description has a cost of $2^{n}$. The repaired GDL $G'$ must ensure that $G' \not\models_{t} \neg \X goal(p_{1},100)$, and the cost of the repair must be no more than $2^{n}-1$. Note that the $\mrpb$ $\langle G, \Phi^{+}, \Phi^{-},cost, C \rangle $ can be constructed in PTIME w.r.t. the size of $\Psi$.

    We define $\sigma:X \rightarrow \{\top,\bot\}$ as a function that assigns each variable $x_{i} \in X$ (i.e., $1 \leq i \leq n$) to be $\top$ or $\bot$, and let $\Psi_{\sigma}$ denotes the truth value of $\Psi$ given the assignment $\sigma$. 
    
    We claim that for any assignment $\sigma$ of $\Psi$, $\Psi_{\sigma}$ is true iff the repair $\mathcal{R}_{\sigma}=\{\langle i, (c, legal(p_{i},pos))\rangle|\sigma(x_{i})=\top\}$
    ensures that $rep(G,\mathcal{R_{\sigma}}) \not\models_{t}  \neg \X goal(p_{1},100)$.
    
    First, it is trivial that the cost of any repaired set $\mathcal{R_{\sigma}}$ defined above is no more than $2^{n}-1$ for any $\sigma$, and in $rep(G,R_{\sigma})$ all players $p_{1},\ldots,p_{n}$ have exactly 1 legal action at $\Sinit$ (either $pos$ or $neg$ but not both).
    
    $\Rightarrow$. If $\sigma$ is a satisfiable assignment. Then, each clause $C_{i}=l_{i}^{1} \vee l_{i}^{2} \vee l_{i}^{3}$ has at least 1 literal satisfied. Wlog assume that $l_{i}^{1}$ is the literal that is satisfied. We further assume that $l_{i}^{1}=x_{k}$ is a positive literal and $\sigma(x_{k})=\top$ (the proof of the case when $l_{i}^{1}=\neg x_{k}$ and $\sigma(x_{k})=\bot$ is similar). Then, in $rep(G,\mathcal{R_{\sigma}})$, we have that $hd(r_{k})'=legal(p_{k},pos)$. Hence, $pos$ is the only legal action for the $p_{k}$ at $\Sinit$. Due to the way $R_{2}$ to $R_{4}$ are defined, the base proposition $c_{i}$ must be justified in the next state because there must be a rule $next(c_{i}) \texttt{:-} does(p_{k},pos).$ in $R_{2}$. Note that the same claim applies to all base propositions $c_{1},\ldots,c_{m}$. Thus, $goal(p_{1},100)$ is justified in the next state. As a result, there is a sequence $(\Sinit,S_{1})$ that has $rep(G,\mathcal{R_{\sigma}}),(\Sinit,S_{1}) \models_{t}  \X goal(p_{1},100)$ which is equivalent to $rep(G,\mathcal{R_{\sigma}}),(\Sinit,S_{1}) \not\models_{t}  \neg\X goal(p_{1},100)$, which means that $rep(G,\mathcal{R_{\sigma}}) \not\models_{t}  \neg\X goal(p_{1},100)$.

    $\Leftarrow$. Let $\sigma$ be an assignment to $\Psi$ If $\mathcal{R_{\sigma}}$ is a solution repair to the $\mrpb$. Then, $rep(G,\mathcal{R_{\sigma}}) \not\models_{t}  \neg \X goal(p_{1},100)$. This is equivalent to $goal(p_{1},100)$ holds at $S_{1}$ for some valid sequence $(\Sinit,S_{1})$. Since there is only 1 rule with $goal(p_{1},100)$ in the head, if $goal(p_{1},100)$ is justified, all $true(c_{1}),\ldots,true(c_{m})$ must be justified. Since we cannot modify the initial states, the only case for $true(c_{i})$ holds is that some of the next rules containing $next(c_{i})$ are justified. In other words, one of the rules of the form $next(c_{i}) \texttt{:-}\mu(l_{i}^{j})$ (for $1 \leq j \leq 3$) must be justified.  
    Wlog, we assume the first one is justified, then it means that $\mu(l_{i}^{1})$ must hold at $S_{0}$. We further assume that $\mu(l_{i}^{1})=does(p_{k},pos)$ (the case when $\mu(l_{i}^{1})=does(p_{k},neg)$ is similar), then we know that $l_{i}^{1}=x_{k}$ and $\langle i,(c,legal(p_{k},pos)) \rangle \in \mathcal{R_{\sigma}}$. This means that $\sigma(x_{k})=\top$, which satisfies the first literal in clause $C_{i}$. Note that the same claim applies to all clauses $C_{1},\ldots,C_{m}$, hence $\sigma$ is a satisfiable assignment to $\Psi$.

    Since SAT is NP-hard, so does $\mrpb$ with $\Phi^{+}=\{\}$.

    % Since $\Psi$ is satisfiable, we know that the optimal repair must have a cost $\leq 2^{n}-1$. Because of the way the cost function is defined, the optimal repair is unique. Since we have just shown the correspondence between the satisfiable assignments of $\Psi$ and solution repairs with cost $\leq 2^{n}-1$, we can claim that $x_{n}=\top$ is in the lexicographically smallest truth assignment of $\Psi$ if and only if the change tuple $\langle n,(c,legal(p_{n},pos)) \rangle$ is in the optimal repair set $\mathcal{R}$.  Note that this is equivalent to checking if the cost of the optimal repair set is odd (i.e., the $\mrpt$ returns yes) because only $cost(n,(c,legal(p_{n},pos)))=1$ is odd while all other modifications to the game description has an even cost. This completes the hardness part of the proof.
    % % $\bullet$ $\mrpv$ is $co-NP$-complete.
    % (Membership) 
    
    %     We then show that $\mrpt$ with $\Phi^{+}=\{\}$ is $\Delta_{2}^{P}$-complete. (Membership) Since the maximum possible cost of a valid repair set is bounded for any $\gd$ $G$ (Definition~\ref{def:repset} and~\ref{def:cost}) and the answer to $\mrpb$ is monotonic, we can binary search on the cost bound $C$ and solve $\mrpb$ for different $C$. The membership follows directly from the fact that $\mrpb$ with $\Phi^{+}=\{\}$ is NP-complete. 

\end{proof}

\begin{theorem}
    Let $G$ be a valid $\gd$, $\mathcal{M}$ a stable model of $\pgen(G)$, and $\mathcal{X}$ the set of all atoms $lit,ha$ in $\mathcal{M}$. If $q(\vec t\,)$ is a ground atom with a predicate symbol in $\{true,does,$ $legal,next\}$ or appears in $G_{R}$, then for any state $S$ and joint action $A$ over $G$'s base propositions and move domain\/:
    \begin{itemize}
               \item $\Pi^{-1}(\mathcal{X}) \cup G_{R} \cup S^{true} \cup A^{does} \models q(\vec t\,)$ iff $\ginv \cup \mathcal{X} \cup G_{R} \cup S^{true} \cup A^{does} \models q(\vec t\,)$.

   %$\Pi^{-1}(\mathcal{X})  \cup G_{R} \cup S^{true} \models q(\vec t\,)$ iff $\ginv \cup \mathcal{X} \cup G_{R} \cup S^{true} \models q(\vec t\,)$
    \end{itemize}
\end{theorem}
\begin{proof}
    Due to the way $\pgen(G)$ is constructed, $\Pi^{-1}(\mathcal{X}) \cup G_{R}$ must be a valid $\gd$ obtained from $G$ with some valid repair. Eventually, $\Pi^{-1}(\mathcal{X})$ will have the form $r_{1}' \cup r_{2}'\cup \ldots \cup r'_{|G_{C}|}$. And for each $i$, $hd(r'_{i})$ is of form $legal(p,a)$, or $next(f)$, or $\emptyset$.
    There are 4 cases.
    
    First, if $q \in \{true,does\}$, the statements trivially hold.
    
    Next, if $q$ is a predicate symbol appears in $G_{R}$, the statements hold because $G$ is in restricted form, any predicate symbol $q \in G_{R}$ %does not depend on $does$ and
    cannot appear in the head of $\ginv \cup \mathcal{X}$ or $\Pi^{-1}(\mathcal{X})$, whether $q(\vec t\,)$ holds solely depends on $S^{true}$ and $G_{R}$. 
    
    Thirdly, if $q(\vec t\,)=legal(p,a)$ for some $p,a$. Since both $\Pi^{-1}(\mathcal{X}) \cup G_{R}$ and $\ginv \cup \mathcal{X} \cup G_{R}$ are valid $\gd$s, whether $legal(p,a)$ holds or not does not depend on $A^{does}$, showing the original statement is the same as showing $\Pi^{-1}(\mathcal{X})  \cup G_{R} \cup S^{true} \models q(\vec t\,)$ iff $\ginv \cup \mathcal{X} \cup G_{R} \cup S^{true} \models q(\vec t\,)$. Since $\Pi^{-1}(\mathcal{X})$ is also in restricted form, we know that $\Pi^{-1}(\mathcal{X}) \cup S^{true} \cup G_{R} \models q(\vec t\,)$ iff there is some $i$ such that $hd(r_{i}')=legal(p,a)$ and there's no $f$ such that: $true(f) \in S^{true}$ and $not~true(f) \in bd(r_{i}')$, \textbf{or} $true(f) \notin S^{true}$ and $true(f) \in bd(r_{i}')$. This is equivalent to $\hd(i,(ac,(p,a))) \in \mathcal{X}$ for some $i$ and there's no $f$ such that: i) $true(f) \in S^{true}$ and $lit(i,(neg,ba,f)) \in \mathcal{X}$, \textbf{or} ii) $true(f) \notin S^{true}$ and $lit(i,(pos,ba,f)) \in \mathcal{X}$. This is exactly what clauses \eqref{enc:rule16}, \eqref{enc:rule17}, and \eqref{enc:rule20} in Fig.~\ref{fig:program} are modeling. \eqref{enc:rule16} and \eqref{enc:rule17} state that $err\_t(i)$ is justified iff for some $f$ and $i$: i) $true(f) \in S^{true}$ and $lit(i,(neg,ba,f)) \in \mathcal{X}$, \textbf{or} ii) $true(f) \notin S^{true}$ and $ lit(i,(pos,ba,f)) \in \mathcal{X}$. \eqref{enc:rule20} forces $\ginv \cup \mathcal{X} \cup S^{true} \cup G_{R} \models q(\vec t\,)$ iff there's some $i$ such that $\hd(i,(ac,(p,a))) \in \mathcal{X}$ and $err\_t(i)$ is not justified. 
    
    Finally, if $q(\vec t\,)=next(f)$ for some $f$, then $\Pi^{-1}(\mathcal{X}) \cup G_{R} \cup S^{true} \cup A^{does} \models q(\vec t\,)$ iff there is some $i$ such that $hd(r_{i}')=next(f)$ and 1): there's no $g$ such that: $true(g) \in S^{true}$ and $not~true(g) \in bd(r_{i}')$, \textbf{or} $true(g) \notin S^{true}$ and $true(g) \in bd(r_{i}')$, also 2): there's no $(p,a)$ such that: $does(p,a) \in A^{does}$ and $not~does(p,a) \in bd(r_{i}')$, \textbf{or} $does(p,a) \notin A^{does}$ and $does(p,a) \in bd(r_{i}')$. We've already discussed case 1) in the proof of the case when $q(\vec t\,)=legal(p,a)$, case 2) is analogous and ensured by~\eqref{enc:rule18}, ~\eqref{enc:rule19}, and~\eqref{enc:rule21}.
    
\end{proof}
\setcounter{corollary}{0}
\renewcommand{\thelemma}{\Alph{section}\arabic{corollary}}
\begin{corollary}
     Let $G$ be a valid $\gd$, $\phi$ be a GTL formula with $deg(\phi)=n$, $\mathcal{M}$ is a stable model of $\pgen(G)$, and $\mathcal{X}$ is the set of all atoms of $lit$ and $\hd$ in $\mathcal{M}$. Let $G'=\Pi^{-1}(\mathcal{X}) \cup G_{R}$, and  
    %, and $\mathcal{M}$ be a stable model of $D \cup \pgen \cup \Pi(G_{C})$. 
    $\pver(\phi)=\plegal^{n}   \cup \pext^{n}(\ginv \cup G_{R}) \cup \enc(\phi,0) \cup \{\texttt{:-} \eta(\phi,0)\}$. Then, $G' \models_{t} \phi$ iff $\mathcal{X} \cup \pver(\phi)$ has no stable model. % and $G' \not\models_{t} \phi$, otherwise. 
\end{corollary}
\begin{proof}
    First note that by definition~\ref{def:temporal}, we know $\pext^{n}(\ginv \cup G_{R}) \cup \mathcal{X}$ and $\pext^{n}(\ginv \cup G_{R} \cup \mathcal{X})$ are the same program because every instances of $lit$ and $\hd$ does not depend on $true$ or $does$ and won't be modified in the temporal extension procedure. We've shown in Theorem 9 that $G'$ and $\ginv \cup G_{R} \cup \mathcal{X}$ are equivalent. Thus, at every state all atoms that can be part of any GTL formula over $G$ that holds in $G'$ must also hold in $\ginv \cup G_{R} \cup \mathcal{X}$, and for any state $S$, its successor state after making the joint action $A$ in $G'$ and  $\ginv \cup G_{R} \cup \mathcal{X}$ is the same state. Therefore, by Theorem~\ref{theorem:basecase}, the corollary holds.
\end{proof}